\documentclass[onecolumn,pra,superscriptaddress]{revtex4-1}

\usepackage[dvips]{graphicx} 
\usepackage{amsfonts}
\usepackage{amssymb}
\usepackage{amscd}
\usepackage{amsmath}    
\usepackage{enumerate}
\usepackage{epsfig}
\usepackage{subfigure}
\usepackage{xcolor}
\usepackage{amsthm}
\usepackage{framed}
\usepackage{multirow}
\usepackage{tcolorbox}
\usepackage{epstopdf}

\newcommand{\tr}{\mathrm{Tr}}

\newcommand{\bra}[1]{\mbox{$\left\langle #1 \right|$}}
\newcommand{\ket}[1]{\mbox{$\left| #1 \right\rangle$}}

\newcommand{\ketbra}[2]{\mbox{$\left| #1 \right\rangle   \left\langle #2 \right|$}}

\newcommand{\re}{\textrm{Re}}
\newcommand{\comments}[1]{}

\newtheorem{lemma}{Lemma}
\newtheorem{theorem}{Theorem}

\newtheorem{corollary}{Corollary}
\newtheorem{problem}{Problem}

\newtheorem{proposition}{Proposition}
\begin{document}
\title{Operational Interpretation of Coherence in Quantum Key Distribution}

\date{\today}
\author{Jiajun Ma$^{\dag}$}
\affiliation{Center for Quantum Information, Institute for Interdisciplinary Information Sciences, Tsinghua University, Beijing 100084, China}

\author{You Zhou$^{\dag}$}
\affiliation{Center for Quantum Information, Institute for Interdisciplinary Information Sciences, Tsinghua University, Beijing 100084, China}

\author{Xiao Yuan}
\affiliation{Department of Materials, University of Oxford, Parks Road, Oxford OX1 3PH, United Kingdom}

\author{Xiongfeng Ma}
\email{xma@tsinghua.edu.cn}
\affiliation{Center for Quantum Information, Institute for Interdisciplinary Information Sciences, Tsinghua University, Beijing 100084, China}

\begin{abstract}
Quantum coherence was recently formalized as a physical resource to measure the strength of superposition. Based on the resource theory, we present a systematic framework that connects a coherence measure to the security of quantum key distribution. By considering a generic entanglement-based key distribution protocol under the collective attack scenario, we show that the key rate can be quantified by the coherence of the shared bipartite states. This framework allows us to derive the key rate of the BB84 and six-state protocols. By utilizing fine-grained parameters, we also derive the improved key rates of both protocols within the coherence-based framework.
Furthermore, we apply it to a practical issue, detection efficiency mismatch, and obtain an improved result. In conclusion, this framework demonstrates the interplay among coherence, entanglement, and quantum key distribution at a fundamental level.
\end{abstract}

\maketitle
\section{Introduction}
As the notion that captures the quantum superposition between differentiable physical states, quantum coherence represents one of the fundamental features that distinguish quantum mechanics from its classical counterpart~\cite{Born1926,schrodinger1935gegenwartige}. Despite of the early recognition of its importance, quantum coherence was only recently formalized under a rigourous framework of resource theory~\cite{Plenio2014coherence}, which stimulated a rapidly growing research field on quantum coherence, ranging from its mathematical characterizations to its operational interpretations~\cite{streltsov2017coherence}.

The motivation for studying the operational interpretation of quantum coherence is two-folded. First, by linking coherence to the operational advantage of quantum information processing protocols, one can improve existing protocols and derive other ones by exploiting similar mechanisms. Second, the observable phenomenon bestowed by quantum coherence allows one to better understand the boundary between quantum and classical realms, one of the fundamental problems in theoretical physics.

The operational significance of quantum coherence has been recognized in many areas, including quantum metrology~\cite{BRAUNSTEIN1996135}, quantum computing~\cite{nielsen2010quantum},  quantum thermodynamics~\cite{aaberg2014catalytic,lostaglio2015description} and quantum biology~\cite{Huelga2013biology}. With the developed resource theory of coherence, more operational significance of coherence was discovered and quantified~\cite{Hillery2016Coherence,Matera2016Coherent,Marvian2016speed,Napoli2016Robustness, lostaglio2015quantum,Bagan2016Path}. Recently, it was shown that coherence quantifies the amount of unpredictable intrinsic randomness from measuring quantum states~\cite{Yuan15randomness,yuan2016interplay}. Such a relation has been applied to develop source-independent quantum random number generators~\cite{Ma2019Coherence}. Taking a qubit as an example, the state $\ket{\psi_A}=(\ket{0}+\ket{1})/\sqrt{2}$ can yield intrinsic randomness when measured in the $Z$ basis, which is unpredictable to an adversary, Eve. In comparison, the measurement result of $\rho_A=(\ket{0}\bra{0}+\ket{1}\bra{1})/2$ with zero coherence can be fully determined by Eve if she holds the purification of $\rho_A$, that is, $\ket{\psi_{AE}}=(\ket{00}+\ket{11})/\sqrt{2}$.

In this paper, we investigate the significant role of coherence played in quantum key distribution (QKD), via considering the relation between coherence and intrinsic randomness. In the scenario of QKD, two legitimate users, Alice and Bob, intend to share a sequence of private and identical bits, called the secret key. In a QKD security analysis, one always needs to consider two steps in postprocessing. One is information reconciliation that ensures the keys shared by Alice and Bob are identical. The other is privacy amplification that extracts the secure key from the raw key. In general, information reconciliation is a standard classical process, while privacy amplification is determined by the quantum procedures of the protocol. Privacy amplification plays a central role in all security proofs. For example, in the Shor-Preskill security proof~\cite{shor2000simple}, privacy amplification is linked to the phase error correction in an equivalent entanglement protocol~\cite{Lo1999unconditional}. In this paper, we examine the postprocessing in an alternative way. After information reconciliation, the amount of secret key that can be extracted from privacy amplification is essentially determined by the intrinsic randomness that is unknown to Eve. This intrinsic randomness can be quantified by the coherence of the joint system of Alice and Bob. For example, in the entanglement-based version of the BB84 protocol~\cite{bennett1984QKD,bennett1992quantum}, supposing there are only phase errors left, the final state shared by Alice and Bob is a mixture of $(\ket{00}+\ket{11})/\sqrt{2}$ and $(\ket{00}-\ket{11})/\sqrt{2}$. If the phase error rate takes the worst case of $50\%$, the state becomes $(\ket{00}\bra{00}+\ket{11}\bra{11})/2$, which has no coherence in the $Z$ basis, and hence no secret key can be generated.

Following this spirit, we propose a generic security analysis framework for QKD under collective attacks, and we show that the key generation rate is closely related to the amount of coherence within the joint quantum states. To be more specific, we find that the security of the key originates from the coherence of the bipartite quantum state shared by Alice and Bob. Our framework is concise, and is one via which one can derive the final key rate formulas of the BB84 protocol~\cite{bennett1984QKD,shor2000simple} and the six-state protocol~\cite{BD98Six,hoi2001proof}. Moreover, the proposed framework allows one to improve the key rates with fine-grained parameters. Many existing QKD security analyses~\cite{Lo1999unconditional,shor2000simple} are based on entanglement distillation protocols~\cite{Bennett1996Purification}, where entanglement is taken as an essential resource to guarantee the security of the final key. In fact, entanglement is a precondition for secure QKD~\cite{Curty2004Precondition}. Thus, we also discuss the interplay among coherence, entanglement, and QKD security. Later, the potential approach to extend our results from the scenario of collective attacks to the one of coherent attacks will be discussed.

Our paper is organized as follows. In Sec.~\ref{Sec:randomness}, we review the close relation between quantum coherence and intrinsic randomness.  In Sec.~\ref{Sec:framework}, we introduce the security analysis framework based on quantifying coherence, and present an explicit key rate formula related to the coherence of the bipartite state in the key generation basis. In Sec.~\ref{Sec:case}, by applying the framework to the BB84 and six-state protocols, we reproduce the original key rate formulas. Then, in Sec.~\ref{Sec:improve}, with analytical tools well developed under the resource theory of coherence, we improve the key rates of these two protocols by using fine-grained parameters in postprocessing the measurement outcomes. Furthermore, in Sec.~\ref{Sec:mismatch}, we apply the framework to solve a practical issue in QKD, detection efficiency mismatch, where two detectors have nonidentical detection efficiency. The derived key rate shows an advantage over previous analyses. We also discuss the interplay among coherence, entanglement, and QKD security in Sec.~\ref{Sec:Entanglement}. Finally, we conclude our work and discuss future works in Sec.~\ref{Sec:conclude}.

\section{Preliminary: coherence and intrinsic randomness}\label{Sec:randomness}
The resource framework of coherence was recently formalized \cite{Plenio2014coherence}. A comprehensive review of this topic can be found in Ref.~\cite{streltsov2017coherence} and references therein. Here, we briefly review the concepts involved in this paper.

The free state and the free operation are two elementary ingredients in all resource theories. In the context of coherence theory, considering a $d$-dimensional Hilbert space $\mathcal{H}_d$ and a reference (computational) basis ${I} =\{\ket{i}\}_{i=1,2,...,d}$, the free state is the state which is diagonal in the reference basis, i.e., $\delta=\sum_{i=1}^d\delta_i \ket{i}\bra{i}$; the free operation is an incoherent complete positive and trace preserving operation, which cannot generate coherence from incoherent states. Based on this coherence framework, several measures are proposed to quantify the superposition strength of the reference basis, such as relative entropy of coherence~\cite{Plenio2014coherence}, which is defined as
\begin{equation} \label{eq:relcoherence}
C(\rho) = S\left(\rho^{\textrm{diag}}\right)-S(\rho).
\end{equation}
where $\rho^{\textrm{diag}}$ is the diagonal state of $\rho$ in the reference basis, $\sum_i\bra{i}\rho\ket{i}\ket{i}\bra{i}$, and $S(\rho)=-\tr[\rho\log_2(\rho)]$  is the von Neumann entropy of $\rho$.

On the other hand, intrinsic randomness is unpredictable compared with the pseudo-randomness produced by deterministic algorithms. A quantum random number generator (QRNG) serves as an elegant solution to the intrinsic randomness generation, via measuring quantum state in well-designed methods~\cite{Ma2016QRNG,herrero2017quantum}. Under the resource framework of coherence, it was recently observed that the coherence of a quantum state quantifies the extractable intrinsic randomness by measuring it in the reference basis~\cite{Yuan15randomness,yuan2016interplay}.

To be more specific, let $\rho_A$ denote the state of system $A$ that is designed to generate random numbers. Consider a purification of $\rho_A$ as $\ket{\psi}_{AE}$, i.e., $\mathrm{Tr}_E[\ket{\psi}_{AE}\bra{\psi}_{AE}]=\rho_A$ with $\mathrm{Tr}_E$ as the partial trace over system $E$. In a randomness analysis, the purification system $E$ is normally assumed to be held by Eve, who aims at predicting the measurement outcome of $\rho_A$ by manipulating system $E$. Suppose a projective measurement on the $I$ basis is performed on $\rho_A$, then the joint state $\rho_{AE}=\ket{\psi}_{AE}\bra{\psi}_{AE}$ becomes $\rho_{AE}'=\sum_{i} \ket{i}_A\bra{i}\otimes\bra{i}_A\rho_{AE}\ket{i}_A$. Considering the independent and identically distributed (i.i.d.) scenario, the intrinsic randomness that is unpredictable to Eve, denoted by $R(\rho_A)$, is measured by the quantum conditional entropy $S(A|E)_{\rho_{AE}'}$. It is shown to be exactly characterized by the relative entropy of coherence $C(\rho_S)$~\cite{Yuan15randomness,yuan2016interplay},
\begin{equation}\label{Eq:randomcoherence}
R(\rho_A)=S(A|E)_{\rho_{AE}'} = C(\rho_A),
\end{equation}
where $S(A|B)=S(\rho_{AB})-S(\rho_B)$ is the conditional quantum entropy of $\rho_{AB}$. Therefore, the resource theory of coherence provides a useful tool to quantify the intrinsic randomness in the reference basis.

At the end of this section, we clarify notations in the paper for clearness. The $Z$-basis is the reference basis of a qubit, $\{\ket{0},\ket{1}\}$. The $X$ basis $\{\ket{+},\ket{-}\}$ and $Y$ basis $\{\ket{+i},\ket{-i}\}$ are conjugate bases with $\ket{\pm}=(\ket{0}\pm\ket{1})/\sqrt{2}$ and $\ket{\pm i}=\frac1{\sqrt{2}}(\ket{0}\pm i\ket{1})$, respectively. Denote the $Z$-basis measurement result as $\mathcal{Z}$. The $Z$-basis parity projectors on a two-qubit space are
\begin{equation}\label{Eq:Zprojector}
\begin{aligned}
\Pi^+&=\ketbra{00}{00}+\ketbra{11}{11}, \\
\Pi^-&=\ketbra{01}{01}+\ketbra{10}{10}.
\end{aligned}
\end{equation}
Similarly, for the $X$-basis and $Y$-basis, the projectors are
\begin{equation}\label{Eq:XYprojector}
\begin{aligned}
\Pi^+_x&=\ket{++}\bra{++}+\ket{--}\bra{--}, \\
\Pi^-_x&=\ket{+-}\bra{+-}+\ket{-+}\bra{-+}, \\
\Pi^+_y&=\ket{+i-i}\bra{+i-i}+\ket{-i+i}\bra{-i+i}, \\ \Pi^-_y&=\ket{+i+i}\bra{+i+i}+\ket{-i-i}\bra{-i-i}.
\end{aligned}
\end{equation}
Moreover, we define the partial dephasing channel on the $Z$ basis as
\begin{equation}\label{Eq:ParDphase}
\Phi(\rho)=\Pi^+\rho\Pi^++\Pi^-\rho\Pi^-.
\end{equation}


\section{Security framework with coherence}\label{Sec:framework}
In this section, we provide a framework that relates the security analysis of QKD to the resource theory of coherence. In QKD, Alice and Bob intend to share a sequence of private and identical bits, called secret key, via communication over an untrusted quantum channel and an authenticated classical channel. Eve can attack the communication channels with any strategies allowed by the principles of quantum mechanics.

In the following discussions, we consider an entanglement-based QKD scheme, since the prepare-and-measure schemes can be converted to the entanglement-based ones with the standard technique~\cite{shor2000simple}. Also, we consider the security analysis with respect to the condition that the shared states between Alice and Bob of different rounds are i.i.d.. This is the collective attack scenario considered in QKD~\cite{Devetak2005distillation}. Then, a generic QKD protocol can be described by the five points below.

\begin{enumerate} [(i)]
\item
$N$ pairs of qubit states, $\rho_{AB}^{\otimes N}$, are distributed to Alice and Bob.
\item
They randomly sample $N-n$ copies of $\rho_{AB}$ for parameter estimation, where $0 < n < N$.
\item
For the remaining $n$ copies of $\rho_{AB}$, Alice and Bob each performs the $Z$-basis measurement to generate the raw key.
\item
They perform classical information reconciliation on the raw key to share identical keys.
\item
They perform privacy amplification based on the information provided in the parameter estimation to share private keys.
\end{enumerate}

In a security proof, the parameter estimation is a crucial step, which determines the amount of secure keys that can be extracted in QKD. Essentially, Alice and Bob perform some measurement $\Gamma_i\in \mathbf{\Gamma}$ to estimate the information of $\rho_{AB}$, with the probability of measurement outcome $i$ being $\gamma_i=\mathrm{Tr}(\rho_{AB}\Gamma_i)$. As a result, $\rho_{AB}$ should fulfill a set of constraints, $\rho_{AB}\in\mathbf{S}$, where $\mathbf{S}$ denotes the set which contains all the states satisfying constraints from parameter estimation,
\begin{equation}\label{Eq:SetS}
\mathbf{S}:=\{\rho_{AB}|\mathbf{\Gamma}: \mathrm{Tr}(\rho_{AB}\Gamma_i)=\gamma_i\}.
\end{equation}

Now we provide the main result of this paper, which connects the key rate of QKD with the relative entropy of coherence. Our derivation is based on the close relation between intrinsic randomness and quantum coherence.

\begin{theorem}\label{Th:keyrate}
In the asymptotic limit where $n, N\rightarrow\infty$, the secret key rate of the above QKD protocol is given by
\begin{equation}\label{Eq:KR}
K = \min_{\rho_{AB}\in\mathcal{S}}C(\Phi(\rho_{AB}))-I_{ec},
\end{equation}
where $\Phi(\cdot)$ is the partial dephasing operation defined in Eq.~\eqref{Eq:ParDphase}, $C(\cdot)$ is the relative entropy of coherence on the computational basis $Z_A\otimes Z_B$ defined in Eq.~\eqref{eq:relcoherence}, and $I_{ec}$ is the information leakage during key reconciliation.
\end{theorem}

Note that $I_{ec}$ in Eq.~\eqref{Eq:KR}, which accounts for the private key consumed in the information reconciliation step, can be directly estimated by the measurement statistics from parameter estimation. Sometimes parameter estimation is not even needed here as long as an error verification step is applied after information reconciliation~\cite{Fung2010Finite}. Thus, the minimization is only on the first term that quantifies the security of the key by quantum coherence. Without loss of generality, in the following analysis, we employ the one-way information reconciliation scheme such that $I_{ec}=H(\mathcal{Z}_A|\mathcal{Z}_B)$ \footnote{Here, one-way information reconciliation means that in the post-processing stage, Alice (or Bob) could determine the final key from her (or his) sifted key directly. For example, one can use Alice's sifted key after privacy-amplification hashing as the final key.}. Here, the Shannon entropy of a random variable $a$ and the conditional entropy of two random variables $a$ and $b$, are denoted by $H(a)=-\sum_a q(a)\log_2 q(a)$ and $H(a|b)=H(ab)-H(b)$, respectively, where $q(a)$ is the underlying probability distribution and $H(e)=-e\log_2 e-(1-e)\log_2(1-e)$. We need to emphasize that our result can be applied to more general postprocessing protocols, e.g., two-way classical-communication protocol~\cite{gottesman2003proof}. This is possible because our framework entirely decouples the privacy amplification step from the information reconciliation step. In the following proof, we show an equivalent virtual protocol which employs quantum bit error correction that commutes with the $Z$-basis measurement. This follows the same spirit of the Lo-Chau and Shor-Preskill proofs for the BB84 protocol~\cite{Lo1999unconditional,shor2000simple}.

\begin{proof}
Let $K(\rho_{AB})$ denote the key rate when the shared quantum state is known to be $\rho_{AB}$. To estimate the secret key rate $K$, one should consider \emph{the worst case} of $\rho_{AB}\in\mathcal{S}$, i.e.
\begin{equation}\label{Eq:KM}
K=\min_{\rho_{AB}\in\mathcal{S}}K(\rho_{AB}),
\end{equation}
where $\mathcal{S}$ is the set of quantum states $\rho_{AB}$ that is consistent with the measurement statistics obtained in the parameter estimation step, as defined in Eq.~\eqref{Eq:SetS}.

First, we consider an equivalent virtual protocol, where Alice and Bob perform the information reconciliation before the $Z$-basis measurement, i.e., steps (iii) and (iv) in Box~1 are exchanged. Then, step (iii) and step (iv) are replaced by

\begin{enumerate}
\item[(iii')]
With the $Z$-basis measurement results obtained in parameter estimation, Alice and Bob perform quantum bit error correction on the $n$ copies of $\rho_{AB}$, which is equivalent to applying a global $Z$-basis parity projector $\{\Pi^+,\Pi^-\}$  on the joint state. Then, Alice (or Bob) applies the $\sigma_x=\ket{0}\bra{1}+\ket{1}\bra{0}$ to rotate all the joint states to the subspace $\Pi^+$.

\item[(iv')]
Alice and Bob perform the $Z$-basis measurement on the error corrected state to generate the identical key.
\end{enumerate}

Note that the quantum bit error correction in step (iii') commutes with the $Z$-basis measurement, since all operations are performed in the $Z$ basis. Thus, this virtual protocol is equivalent to the one shown in Box~1. The quantum bit error correction in the virtual protocol can be realized with pre-shared $n H(\mathcal{Z}_A|\mathcal{Z}_B)$  Einstein-Podolsky-Rosen (EPR) pairs. In the original protocol, the amount of key cost is given by the conditional entropy $H(\mathcal{Z}_A|\mathcal{Z}_B)$. This step is essentially classical. See Appendix \ref{App:QEC} for more detailed discussions. We define the bit error rate $e_b=\mathrm{Tr}(\Pi^-\rho_{AB})$, and
\begin{equation}\label{eq:ECpart}
H(\mathcal{Z}_A|\mathcal{Z}_B)\le H(e_b),
\end{equation}
where the equality holds for the symmetric case.

After the quantum bit error correction step (iii'), the original $\rho_{AB}^{\otimes n}$ is transformed to $n(1-e_b)$ copies of $\rho_{AB}^+=\Pi^+\rho_{AB}\Pi^+/\mathrm{Tr}(\Pi^+\rho_{AB})$ and $ne_b$ copies of $\sigma_x^B\rho_{AB}^-\sigma_x^B$, with $\rho_{AB}^-=\Pi^-\rho_{AB}\Pi^-/\mathrm{Tr}(\Pi^-\rho_{AB})$. In step (iv'), when Alice and Bob measure these states in the $Z$ basis, they will get identical keys.


To perform the privacy amplification in step (v), one essentially needs to characterize the amount of intrinsic randomness in the error corrected keys that are unpredictable to Eve. Thus the key rate shows
\begin{equation}\label{Eq:K1}
K(\rho_{AB})=\frac1n\left\{n(1-e_b)R(\rho_{AB}^+)+ne_b R(\sigma_x^B\rho_{AB}^-\sigma_x^B)-nH(\mathcal{Z}_A|\mathcal{Z}_B)\right\}.
\end{equation}
Recall the relation between intrinsic randomness and coherence shown in Eq.~\eqref{Eq:randomcoherence},
\begin{equation}\label{Eq:RC}
R(\rho)=C(\rho),
\end{equation}
where the reference basis of relative entropy of coherence $C$ coincides with the basis $Z_A\otimes Z_B$. Then we have
\begin{equation}\label{Eq:3K12}
\begin{aligned}
K(\rho_{AB})&=\frac1n\left\{n(1-e_b)C(\rho_{AB}^+)+ne_b C(\sigma_x^B\rho_{AB}^-\sigma_x^B)-nH(\mathcal{Z}_A|\mathcal{Z}_B)\right\} \\
&=(1-e_b)C(\rho_{AB}^+)+e_b C(\rho_{AB}^-)-H(\mathcal{Z}_A|\mathcal{Z}_B) \\
&=C((1-e_b)\rho_{AB}^++e_b\rho_{AB}^-)-H(\mathcal{Z}_A|\mathcal{Z}_B) \\
&=C(\Phi(\rho_{AB}))-H(\mathcal{Z}_A|\mathcal{Z}_B),
\end{aligned}
\end{equation}
where the third equality employs the additivity property of coherence and the Hilbert space of $\rho_{AB}^+$ and $\rho_{AB}^-$ are orthogonal~\cite{yu2016alternative}. Inserting Eq.~\eqref{Eq:3K12} into Eq.~\eqref{Eq:KM}, one obtains Eq.~\eqref{Eq:KR}.
\end{proof}

Note that in the symmetric case, where the bit value of the raw key is evenly distributed, the information reconciliation part is given by Eq.~\eqref{eq:ECpart} with equality, then the key rate formula can be further written as,
\begin{equation}\label{Eq:KReb}
K = \min_{\rho_{AB}\in\mathcal{S}}C(\Phi(\rho_{AB}))-H(e_b).
\end{equation}
In general, for the asymmetric case, $H(e_b)\geq H(\mathcal{Z}_A|\mathcal{Z}_B)$ on account of Fano's inequality.

\section{Key rates of BB84 and six-state QKD}\label{Sec:case}
As examples, we apply the framework to the security analysis of the BB84 and six-state QKD protocols in the collective-attack scenario. One can see that the secret key rates of these two protocols can be directly derived with the tools well developed within the resource theory of coherence~\cite{streltsov2017coherence}. We list the results of these two re-derivations as the following corollaries. Note that these two protocols only differ from each other on the measurement $\{\Gamma_i\}$ performed in parameter estimation. For simplicity, we consider the symmetric case, where Eq.~\eqref{Eq:KReb} holds.

\subsection{BB84 protocol}
Consider the entanglement-based version of BB84 protocol, where in parameter estimation, Alice and Bob obtain the bit error rate $e_b=\mathrm{Tr}(\Pi^-\rho_{AB})$ and the phase error rate $e_p=\mathrm{Tr}(\Pi^-_x\rho_{AB})$. Then following Theorem \ref{Th:keyrate}, we have the following corollary.

\begin{corollary}\label{Cr:BB84}
The key rate of the BB84 QKD protocol $K_{BB84}$ is given by
\begin{equation} \label{Eq:reBB84}
\begin{aligned}
K_{BB84} &= \min_{\rho_{AB}\in\mathcal{S}}C(\Phi(\rho_{AB}))-H(e_b) \\
&= 1- H(e_p)-H(e_b),
\end{aligned}
\end{equation}
where $\mathcal{S}$ contains all the states yielding the same bit error rate $e_b$ and phase error rate $e_p$ obtained from parameter estimation.
\end{corollary}

The result is consistent with the one from the Shor-Preskill security proof~\cite{shor2000simple}. We prove this corollary by first showing that $K(\rho_{AB})\geq K_{BB84}$ for all $\rho_{AB}\in \mathcal{S}$, and then giving a specific state in this set to saturate the inequality.

\begin{proof}
First note that the four eigenstates of $Z_A\otimes Z_B$ and $X_A\otimes X_B$ are a pair of mutually unbiased bases in the four-dimensional system $H_A^2\otimes H_B^2$. Denote $\Delta_{Z _{AB}}$ ($\Delta_{X _{AB}}$) to be the projective measurement outcome on the $Z_A\otimes Z_B$ ($X_A\otimes X_B$) basis. Due to the entropy uncertainty relation~\cite{maassen1988generalized,berta2010uncertainty}, for any state $\rho$, we have
\begin{equation}\label{}
H(\Delta_{Z _{AB}}(\rho))+H(\Delta_{X _{AB}}(\rho))\geq 2+S(\rho).
\end{equation}
Hence the relative entropy of coherence in the $Z$ basis satisfies~\cite{Ma2019Coherence}
\begin{equation}\label{Eq:coh:uncertain}
C_{Z _{AB}}(\rho)=H(\Delta_{Z _{AB}}(\rho))-S(\rho)\geq 2-H(\Delta_{X _{AB}}(\rho)).
\end{equation}
Denoting the rank-2 projective measurement $\{\Pi^+_x,\Pi^-_x\}$ outcomes by $\Delta_{XX}$, one has the key rate in Eq.~\eqref{Eq:KReb},
\begin{equation} \label{Eq:K12}
\begin{aligned}
K_{BB84}(\rho_{AB})&=C(\Phi(\rho_{AB}))-H(e_b) \\
&\geq 2-H(\Delta_{X _{AB}}(\Phi(\rho_{AB})))-H(e_b) \\
&=1-H(\Delta_{XX}(\Phi(\rho_{AB})))-H(e_b), \\
&=1-H(e_p)-H(e_b)
\end{aligned}
\end{equation}
where Eq.~\eqref{Eq:coh:uncertain} is applied for state $\Phi(\rho_{AB})$ in the second line. The third line holds due to the following reason.
For the state $\Phi(\rho_{AB})$ which is the partially dephased state on $\Pi^{+}$ and $\Pi^{-}$ subspaces, it satisfies,
\begin{equation}
\begin{aligned}
\bra{++}\Phi(\rho_{AB})\ket{++}&=\bra{--}\Phi(\rho_{AB})\ket{--}=\frac{1-e_p}{2},\\
\bra{+-}\Phi(\rho_{AB})\ket{+-}&=\bra{-+}\Phi(\rho_{AB})\ket{-+}=\frac{e_p}{2}.
\end{aligned}
\end{equation}
That is, it has equal probabilities inside each of the rank-2 projectors of the $X$ basis, thus $H(\Delta_{X _{AB}}(\Phi(\rho_{AB})))=1+H(\Delta_{XX}(\Phi(\rho_{AB})))=1+H(e_p)$.

Finally, one can see that the Bell diagonal state, with probabilities on $\{\ket{\Phi^+}, \ket{\Phi^-}, \ket{\Psi^+}, \ket{\Psi^-}\}$ being $\{(1-e_b)(1-e_p),(1-e_b)e_p, e_b(1-e_p), e_be_p\}$, reaches the minimal key rate $K_{BB84}$ in the state set $\mathcal{S}$.
\end{proof}

\subsection{Six-state protocol}
Consider the entanglement-based six-state protocol, where in parameter estimation, Alice and Bob perform the measurement in the $X$, $Y$ and $Z$ basis respectively. Then, they estimate the error in these three basis $e_x=e_p$, $e_y=\mathrm{Tr}(\Pi^-_y\rho_{AB})$, and $e_z=e_b$. Hence we have three parameters $e_x, e_y$ and $e_z$ to constrain the state $\rho_{AB}$.

Suppose that the diagonal terms of $\rho_{AB}$, when represented in the Bell diagonal basis $\{\ket{\Phi^+}, \ket{\Phi^-}, \ket{\Psi^+}, \ket{\Psi^-}\}$, is $\{p_0, p_1, p_2, p_3\}$ with $p_i\geq 0$ and $\sum_i p_i=1$. Note that these $p_i$ are directly related to the estimated error rates, i.e.,
\begin{eqnarray}
e_x &= p_1 + p_3,\label{Eq:six:ex}\\
e_y &= p_1 + p_2,\label{Eq:six:ey}\\
e_z &= p_2 + p_3.\label{Eq:six:ez}
\end{eqnarray}
Then following Theorem \ref{Th:keyrate}, we have the following corollary.

\begin{corollary}\label{Cr:six}
The key rate of the six-state QKD protocol $K_{six}$ is given by,
\begin{equation} \label{Eq:resix}
\begin{aligned}
K_{six} &= \min_{\rho_{AB}\in\mathcal{S}}C(\Phi(\rho_{AB}))-H(e_b) \\
&= 1-H(\{p_i\}),
\end{aligned}
\end{equation}
where $\mathcal{S}$ contains all the states yielding the same error rates $e_x$,  $e_y$, and $e_z$ obtained from parameter estimation.
\end{corollary}

The result is consistent with the one from the previous security proof~\cite{hoi2001proof}. Note that the state set $\mathcal{S}$ is more restrained compared to the one in the BB84 protocol. We prove this corollary by first showing that $K(\rho_{AB})\geq K_{six}$ for all $\rho_{AB}\in \mathcal{S}$, and then giving a specific state in this set to saturate the inequality.

\begin{proof}
Considering $\sum_i p_i=1$, with Eqs.~\eqref{Eq:six:ex} to \eqref{Eq:six:ez} $\{p_i\}$ can be estimated by
\begin{equation} \label{Eq:six_para}
\begin{aligned}
p_0 &= \frac{2-e_x-e_y-e_z}{2},\\
p_1 &= \frac{e_x + e_y - e_z}{2},\\
p_2 &= \frac{e_y + e_z - e_x}{2},\\
p_3 &= \frac{e_z + e_x - e_y}{2}.
\end{aligned}
\end{equation}
Applying Eq.~\eqref{Eq:KReb}, the key rate is given by
\begin{equation} \label{Eq:six:re}
\begin{aligned}
K_{six}(\rho_{AB})&=C(\Phi(\rho_{AB}))-H(e_z) \\
&=(1-e_z)C(\rho_{AB}^+)+e_z C(\rho_{AB}^-)-H(e_z) \\
&\geq(1-e_z)\left[1-H\left(\frac{p_0}{p_0+p_1}\right)\right] +e_z\left[1-H\left(\frac{p_2}{p_2+p_3}\right)\right]-H(e_z) \\
&=1-(1-e_z)H\left(\frac{p_0}{p_0+p_1}\right)-e_zH\left(\frac{p_2}{p_2+p_3}\right)-H(e_z) \\
&=1-H(\{p_i\}),
\end{aligned}
\end{equation}
where in the last line we substitute the relation of $e_z$ in Eq.~\eqref{Eq:six:ez}. The third line can be derived as follows. For the $Z$ and $X$ bases, two mutually unbiased bases of a qubit, the uncertainty relation for coherence measures is given by~\cite{Ma2019Coherence}
\begin{equation}\label{Eq:coh_unc}
C_{Z}(\rho)=H(\Delta_{Z }(\rho))-S(\rho)\geq 1-H(\Delta_{X}(\rho)).
\end{equation}
Since $\rho_{AB}^+$ is rank-2, it can be viewed as a ``qubit" state and the corresponding $Z$ and $X$ bases are $Z'=\{\ket{00},\ket{11}\}$ and $X'=\{\ket{\Phi^+}, \ket{\Phi^-}\}$ respectively, where $\ket{\Phi^{\pm}}=(\ket{00}\pm\ket{11})/\sqrt2$. Thus, applying Eq.\eqref{Eq:coh_unc} to $\rho_{AB}^+$, one has
\begin{equation}\label{Eq:six:plus}
C_{Z'}(\rho_{AB}^+)\geq 1-H(\Delta_{X'}(\rho_{AB}^+))=1-H\left(\frac{p_0}{p_0+p_1}\right).
\end{equation}
Similarly,
\begin{equation}\label{Eq:six:minus}
C_{Z''}(\rho_{AB}^-)\geq 1-H(\Delta_{X''}(\rho_{AB}^-))=1-H\left(\frac{p_2}{p_2+p_3}\right),
\end{equation}
where $Z''$ and $X''$ bases are $\{\ket{01},\ket{10}\}$ and $\{\ket{\Psi^+}, \ket{\Psi^-}\}$ respectively, with $\ket{\Psi^{\pm}}=(\ket{01}\pm\ket{10})/\sqrt2$. Based on Eq.~\eqref{Eq:six:plus} and  Eq.~\eqref{Eq:six:minus}, we obtain the inequality in the third line of Eq.~\eqref{Eq:six:re}.

Finally, it is straightforward to verify that the Bell diagonal state with probabilities $\{p_0, p_1, p_2, p_3\}$ reaches the minimal key rate $K_{six}$ in the state set $\mathcal{S}$.
\end{proof}

\section{Improve the key rate with the framework}\label{Sec:improve}
In this section, we show that the above security proof framework allows us to improve the key rates using fine-grained parameters obtained in the measurement outcomes. Essentially, if one can acquire more information about $\rho_{AB}$ in the parameter estimation step, then the state set $\mathcal{S}$ in Eq.~\eqref{Eq:KR} will be constrained more tightly, and the key rate can be improved.

Here we point out an important observation about Theorem \ref{Th:keyrate}. In order to estimate the secret key rate generated by an unknown $\rho_{AB}$, it suffices to gain the information of $\Phi(\rho_{AB})$, rather than the full information of $\rho_{AB}$. To be more specific, Alice and Bob only need to estimate
\begin{equation}\label{Eq:DephaseD}
\Phi(\rho_{AB})=
\left(
  \begin{array}{cccc}
    m_{00} & 0 & 0 & m_{03} \\
    0 & m_{11} & m_{12} & 0 \\
    0 & m_{21} & m_{22} & 0 \\
    m_{30} & 0 & 0 & m_{33} \\
  \end{array}
\right),
\end{equation}
where $m_{ij}$ are the density matrix elements of $\rho_{AB}$ in the $Z$ basis, $\sum_i m_{ii}=1$, $m_{12}=m_{21}^*$ and $m_{03}=m_{30}^*$. The form in Eq.~\eqref{Eq:DephaseD} provides clear clues to improve the key rates. In the following two subsections, we show the refined key rates for BB84 and six-state protocols with the tools from the resource theory of coherence.

\subsection{BB84 protocol}\label{Sec:imBB84}
In the BB84 protocol, the relations between the error rates $e_b$, $e_p$ and the matrix elements of $\rho_{AB}$, as shown in Eq.\eqref{Eq:DephaseD}, are
\begin{eqnarray}
e_b&=&m_{11}+m_{22},\\
e_p&=&1/2-\re[m_{03}]-\re[m_{12}].\label{perror:BB84}
\end{eqnarray}
In parameter estimation, Alice and Bob carry out $Z_A$ and $Z_B$ measurement, respectively. Thus from the measurement results they can obtain not only the bit error rate $e_b$, but also the four diagonal elements in the $Z_A\otimes Z_B$ basis, i.e., $m_{00}$, $m_{11}$, $m_{22}$, $m_{33}$. Hence the bit error rate $e_b$ can be seen as a coarse-grained parameter from the four diagonal elements.

Based on this observation, we give the refined key rate formula for the BB84 protocol. First, let us define the following optimization problem.
\begin{problem}\label{Pro:1}
Minimize $C(\rho(a,b))$ that is subject to  $a+b=1/2-e_p$, $|a|\le \sqrt{m_{00}m_{33}}$ and $|b|\le \sqrt{m_{11}m_{22}}$ with $a, b\in \mathbb{R}$,  where $C$ is the relative entropy of coherence, and
\begin{equation}\label{Eq:Pro:1}
\rho(a,b)=
\left(
  \begin{array}{cccc}
    m_{00} & 0 & 0 & a\\
    0 & m_{11} & b & 0 \\
    0 & b & m_{22} & 0 \\
    a & 0 & 0 & m_{33} \\
  \end{array}
\right).
\end{equation}
\end{problem}
This optimization problem can be efficiently solved via simple numerical methods. In addition, when the diagonal elements satisfy $m_{00}/m_{33}=m_{11}/m_{22}$ (or $m_{00}/m_{33}=m_{22}/m_{11}$), it can be analytically solved, as shown in Lemma \ref{Lem:sym:pro1} in Appendix \ref{sym:pro1}. We have the following theorem to improve the key rate of the BB84 protocol.

\begin{theorem}\label{Th:KR:BB84Opt}
The secret key rate of the BB84 QKD protocol can be estimated via
\begin{eqnarray}\label{Eq:KR:BB84Opt}
K_{BB84}^{opt}=C(\rho(\bar{a}, \bar{b}))-H(e_b),
\end{eqnarray}
where $\{\bar{a}, \bar{b}\}$ is the solution to Problem \ref{Pro:1}.
\end{theorem}

\begin{proof}
From Eq.~\eqref{Eq:KReb}, we need to prove that
\begin{equation} \label{Eq:BB84Opt:Der}
\begin{aligned}
K &= \min_{\rho_{AB}\in\mathcal{S}}C(\Phi(\rho_{AB}))-H(e_b) \\
&=K_{BB84}^{opt},
\end{aligned}
\end{equation}
where $\mathcal{S}$ contains all the states sharing the same diagonal elements $m_{00}$, $m_{11}$, $m_{22}$, $m_{33}$ and the phase error rate $e_p$ obtained from parameter estimation.

Define $\sigma_{AB}$ as the state by removing the imaginary parts of the off-diagonal terms in $\Phi(\rho_{AB})$,
\begin{equation}\label{Eq:sigma}
\sigma_{AB}=
\left(
  \begin{array}{cccc}
    m_{00} & 0 & 0 & \re[m_{03}] \\
    0 & m_{11} & \re[m_{12}] & 0 \\
    0 & \re[m_{21}] & m_{22} & 0 \\
    \re[m_{30}] & 0 & 0 & m_{33} \\
  \end{array}
\right).
\end{equation}
It is clear that $C(\Phi(\rho_{AB}))\ge C(\sigma_{AB})$, due to the fact that the magnitude of the off-diagonal elements is reduced. Specifically, considering a qubit density matrix,
\begin{equation}\label{}
\rho=
\left(
  \begin{array}{cc}
   \beta & |c|e^{i\varphi}\\
    |c|e^{-i\varphi}& 1-\beta  \\
  \end{array}
\right),
\end{equation}
after applying the incoherent operation $\hat{O}_r(\rho)=\frac1{2}U\rho U^\dag+\frac1{2}\rho$ with $U=\ket{0}\bra{0}+e^{2i\varphi}\ket{1}\bra{1}$, one can get,
\begin{equation}\label{}
\hat{O}_r(\rho)=
\left(
  \begin{array}{cc}
   \beta & |c|\cos(\varphi)\\
    |c|\cos(\varphi) & 1-\beta \\
  \end{array}
\right),
\end{equation}
where the imaginary parts of the off-diagonal terms are removed. As coherence does not increase under incoherent operation, $C(\rho)\geq C(\hat{O}_r(\rho))$~\cite{Plenio2014coherence}.

Since $\Phi(\rho_{AB})$ locates in the two rank-2 subspaces $\Pi_+$ and $\Pi_-$, similarly, one can get $C(\Phi(\rho_{AB}))\ge C(\sigma_{AB})$ via applying incoherent operations on these two subspaces respectively.
As a result, for any state $\rho_{AB}\in \mathcal{S}$,
\begin{equation} \label{Eq:BB84:inq}
\begin{aligned}
K(\rho_{AB}) &=C(\Phi(\rho_{AB}))-H(e_b) \\
&\geq C(\sigma_{AB})-H(e_b) \\
&\geq \min_{\sigma_{AB}\in\mathcal{S_{\sigma}}}C(\sigma_{AB})-H(e_b) \\
&=C(\rho(\bar{a}, \bar{b}))-H(e_b),
\end{aligned}
\end{equation}
where $\mathcal{S_{\sigma}}$ consists of all the corresponding $\sigma_{AB}$ from $\Phi(\rho_{AB})$,
and the last line is on account of the definition of Problem \ref{Pro:1}. Note that $\sigma_{AB}$ is also a quantum state belonging to the state set $\mathcal{S}$, i.e., $\sigma_{AB}\in\mathcal{S}$, thus the inequality above can be saturated and we get Eq.~\eqref{Eq:BB84Opt:Der}.
\end{proof}
Now we have the following corollary.

\begin{corollary}\label{Lm:BB84ip}
For the BB84 QKD protocol, $K_{BB84}^{opt}$ in Eq.\eqref{Eq:KR:BB84Opt} generally yields a higher secret key rate than the Shor-Preskill one, $K_{BB84}$ in Eq.~\eqref{Eq:reBB84}:
\begin{equation} \label{Eq:BB84com}
K_{BB84}^{opt}\geq K_{BB84}.
\end{equation}
\end{corollary}

Corollary \ref{Lm:BB84ip} can be directly obtained from Eqs.~\eqref{Eq:reBB84} and Eq.~\eqref{Eq:BB84Opt:Der}. Specifically, since more parameters are utilized to constrain the state $\rho_{AB}$, the state set $\mathcal{S}$ in Eq.~\eqref{Eq:BB84Opt:Der} is the subset of the one in Eq.~\eqref{Eq:reBB84}. As a result, one has $K_{BB84}^{opt}\geq K_{BB84}$. The proof of Corollary \ref{Cr:BB84} is based on the entropy uncertainty relation. Here, we prove Corollary \ref{Lm:BB84ip} with the tools from the coherence theory~\cite{streltsov2017coherence}. In this way, one can clearly see when the inequality in Eq.~\eqref{Eq:BB84com} is saturated.
\begin{proof}
Define $\hat{O}_{ij}$ as the operation
\begin{equation}
\hat{O}_{ij}(\rho)=\frac1{2}S_{ij}\rho S_{ij}+\frac1{2}\rho,
\end{equation}
where $S_{ij}=\ketbra{i}{j}+\ketbra{j}{i}$. Then, we consider the state $\sigma_{AB}'=\hat{O}_{12}\circ \hat{O}_{03}(\sigma_{AB})$, where $\sigma_{AB}$ is defined in Eq.\eqref{Eq:sigma}. Here the labels $\{0,1,2,3\}$ represent the $Z$ basis $\{\ket{00},\ket{01}\ket{10}\ket{11}\}$ respectively. And we have $\sigma_{AB}'$,
\begin{equation}\label{Eq:sigma:p}
\sigma_{AB}'=
\left(
  \begin{array}{cccc}
  \frac{1-e_b}{2} & 0 & 0 & \re[m_{03}]\\
    0 & \frac{e_b}{2} & \re[m_{12}] & 0 \\
    0 & \re[m_{21}] & \frac{e_b}{2} & 0 \\
    \re[m_{30}] & 0 & 0 & \frac{1-e_b}{2}\\
  \end{array}
\right),
\end{equation}
where the diagonal elements of the density matrix become equal in two subspaces $\Pi^+$ and $\Pi^-$ respectively after the operation.
Clearly, $\hat{O}_{ij}$ is an incoherent operation, thus the coherence of $\sigma_{AB}'$ is not larger than that of $\sigma_{AB}$, i.e.,
\begin{equation}\label{}
C(\sigma_{AB})\ge C(\sigma_{AB}').
\end{equation}
By definition, one has
\begin{eqnarray}
K_{BB84}^{opt}&= \min_{\sigma_{AB}\in\mathcal{S_{\sigma}}}C(\sigma_{AB})-H(e_z)\\\nonumber
&\ge \min_{\sigma_{AB}\in\mathcal{S_{\sigma}'}} C(\sigma_{AB}')-H(e_z),
\end{eqnarray}
where $\mathcal{S_{\sigma}'}$ contains all the states $\sigma_{AB}'$ obtained from $\sigma_{AB}$, as shown in Eq.~\eqref{Eq:sigma:p}. In fact, the minimization in the second line is a special case of Problem \ref{Pro:1}. By applying Lemma \ref{Lem:sym:pro1} in  Appendix \ref{sym:pro1}, one can get the minimal value, $1- H(e_p)-H(e_b)$. In the end, we have $K_{BB84}^{opt}\geq K_{BB84}$.
\end{proof}

From Eq.~\eqref{Eq:sigma:p}, it is clear to see that $K_{BB84}^{opt}= K_{BB84}$ when the diagonal elements in the two subspaces $\Pi^+$ and $\Pi^-$ are balanced, i.e., $m_{00}= m_{33}$ and $m_{11}= m_{22}$. In practice, in order to achieve this improvement of the key rate, Alice and Bob need to replace the estimation of $e_b$ with more refined parameters $m_{00}$, $m_{11}$, $m_{22}$, $m_{33}$ in the parameter estimation step, then perform privacy amplification with the updated key rate formula Eq.~\eqref{Eq:KR:BB84Opt}. Note that this modification does not require extra quantum or classical communications between Alice and Bob.

\subsection{Six-state protocol}\label{Sec:im6}
Similar to the case of BB84 protocol, one can improve the key rate of six-state protocol by utilizing more refined parameters, i.e., diagonal elements $m_{00}$, $m_{11}$, $m_{22}$, $m_{33}$, instead of the coarse-grained one, $e_z$. Here, we provide the following theorem.

\begin{theorem}\label{Th:KR:sixOpt}
The secret key rate of six-state QKD protocol can be estimated via
\begin{equation}\label{Eq:KR:sixOpt}
K_{six}^{opt}= C(\tau)-H(e_z),
\end{equation}
where
\begin{equation}\label{Eq:Tau}
\tau=
\left(
  \begin{array}{cccc}
    m_{00} & 0 & 0 & (1-e_x-e_y)/2\\
    0 & m_{11} & (e_y-e_x)/2 & 0 \\
    0 & (e_y-e_x)/2 & m_{22} & 0 \\
    (1-e_x-e_y)/2 & 0 & 0 & m_{33} \\
  \end{array}
\right).
\end{equation}
\end{theorem}
\begin{proof}
The proof is similar to that of Thereom \ref{Th:KR:BB84Opt}. We need to prove that
\begin{equation}\label{Eq:sixOpt:Der}
\begin{aligned}
K &= \min_{\rho_{AB}\in\mathcal{S}}C(\Phi(\rho_{AB}))-H(e_b)\\\nonumber
&=K_{six}^{opt},
\end{aligned}
\end{equation}
where $\mathcal{S}$ contains all the states sharing the same diagonal elements $m_{00}$, $m_{11}$, $m_{22}$, $m_{33}$ and the error rates $e_x$ and $e_y$ obtained from parameter estimation. Recall Eq.\eqref{Eq:BB84:inq},
\begin{equation} \label{Eq:T3}
\begin{aligned}
K &= \min_{\rho_{AB}\in\mathcal{S}}C(\Phi(\rho_{AB}))-H(e_b) \\
&\geq \min_{\sigma_{AB}\in\mathcal{S_{\sigma}}}C(\sigma_{AB})-H(e_z),
\end{aligned}
\end{equation}
where $\sigma_{AB}$ is defined in Eq.\eqref{Eq:sigma}. Here, $\mathcal{S_{\sigma}}=\{\tau\}$ only has one element, since terms in $\sigma_{AB}$ can all be determined by parameter estimation in the six-state protocol. Namely, one has
\begin{eqnarray}
e_x&=&1/2-\re[m_{03}]-\re[m_{12}],\\
e_y&=&1/2-\re[m_{03}]+\re[m_{12}],\\
e_z&=&m_{11}+m_{22},
\end{eqnarray}
while $m_{00}$, $m_{11}$, $m_{22}$, $m_{33}$ can be estimated with the $Z_A\otimes Z_B$ measurement. Inserting $\sigma_{AB}=\tau$ into Eq.~\eqref{Eq:T3} and noting that $\tau\in \mathcal{S}$, we obtain Eq.~\eqref{Eq:sixOpt:Der}.
\end{proof}

\begin{corollary}\label{Lm:sixip}
For the six-state QKD protocol, $K_{six}^{opt}$ in Eq.~\eqref{Eq:KR:sixOpt} generally yields a higher secret key rate than the original one, $K_{six}$ in Eq.~\eqref{Eq:resix}:
\begin{equation}
K_{six}^{opt}\geq K_{six}.
\end{equation}
\end{corollary}
Like in the BB84 case, one can obtain Corollary \ref{Lm:sixip} directly from Eqs.~\eqref{Eq:resix} and \eqref{Eq:sixOpt:Der}. Here we show a proof based on the coherence theory~\cite{streltsov2017coherence}.
\begin{proof}
Similar to the proof in Corollary \ref{Lm:BB84ip}, we apply the incoherent operation $\hat{O}$ on the state $\tau$, and get $\tau'=\hat{O}_{12}\circ \hat{O}_{03}(\tau)$, that is
\begin{equation}\label{Eq:Tau'}
\tau'=
\left(
  \begin{array}{cccc}
  \frac{1-e_z}{2} & 0 & 0 & (1-e_x-e_y)/2\\
    0 & \frac{e_z}{2} & (e_y-e_x)/2 & 0 \\
    0 & (e_y-e_x)/2 & \frac{e_z}{2} & 0 \\
    (1-e_x-e_y)/2 & 0 & 0 & \frac{1-e_z}{2}\\
  \end{array}
\right).
\end{equation}
Due to monotonicity of coherence under incoherent operation, one has
\begin{equation}
\begin{aligned}
K_{six}^{opt}&= C(\tau)-H(e_z) \\
&\ge C(\tau')-H(e_z) \\
&=1-H(\{p_i\}) \\
&=K_{six}.
\end{aligned}
\end{equation}
Here, we substitute the probabilities $p_i$ on the Bell diagonal basis for the error rates $e_x,e_y$, and $e_z$ with Eqs.~\eqref{Eq:six:ex} to \eqref{Eq:six:ez}, and calculate the coherence $C(\tau')$.
\end{proof}

From Eq.~\eqref{Eq:Tau'}, it is clear to see that $K_{six}^{opt}= K_{six}$ when the diagonal elements in the two subspaces $\Pi^+$ and $\Pi^-$ are balanced, i.e. $m_{00}= m_{33}$ and $m_{11}= m_{22}$.
In practice, to achieve this improvement of the key rate, Alice and Bob need to replace the estimation of $e_z$ with the more refined parameters $m_{00}$, $m_{11}$, $m_{22}$, $m_{33}$ in the parameter estimation step, then perform the privacy amplification with the updated key rate formula Eq.~\eqref{Eq:KR:sixOpt}. Note that this modification does not require extra quantum or classical communications between Alice and Bob.

Here, we have some remarks regarding the improvement of the key rates. In Secs.~\ref{Sec:imBB84} and \ref{Sec:im6}, we have shown that under the coherence-based framework, one can improve the key rates of BB84 and the six-state protocol with fine-grained parameters. These key rate improvements can be understood as a fine-grained estimation of coherence in  $\Phi(\rho_{AB})$. On the other hand, by utilizing other key rate formulas, such as the Devetak-Winter approach, one may also get similar improvements of the key rates by using fine-grained parameters. See Appendix \ref{Ap:Winter} for more discussions.

\subsection{Numerical simulation}
To illustrate the improvement on the security analysis via the coherence framework, we numerically compare the four key rates analyzed above in Fig.~\ref{fig:CKR}, i.e., $K_{BB84}$ in Eq.~\eqref{Eq:reBB84}, $K_{BB84}^{opt}$ in Eq.~\eqref{Eq:KR:BB84Opt}, $K_{six}$ in Eq.~\eqref{Eq:resix} and $K_{six}^{opt}$ in Eq.~\eqref{Eq:KR:sixOpt}.
Here we set typical experimental parameters for simulation, and use a parameter $\alpha$ to describe the unbalance of the diagonal elements of $\rho_{AB}$, i.e. $m_{00}/m_{33}=m_{22}/m_{11}=\frac{\alpha}{1-\alpha}$.

\begin{figure}[bht]
\centering
\resizebox{8cm}{!}{\includegraphics[scale=1]{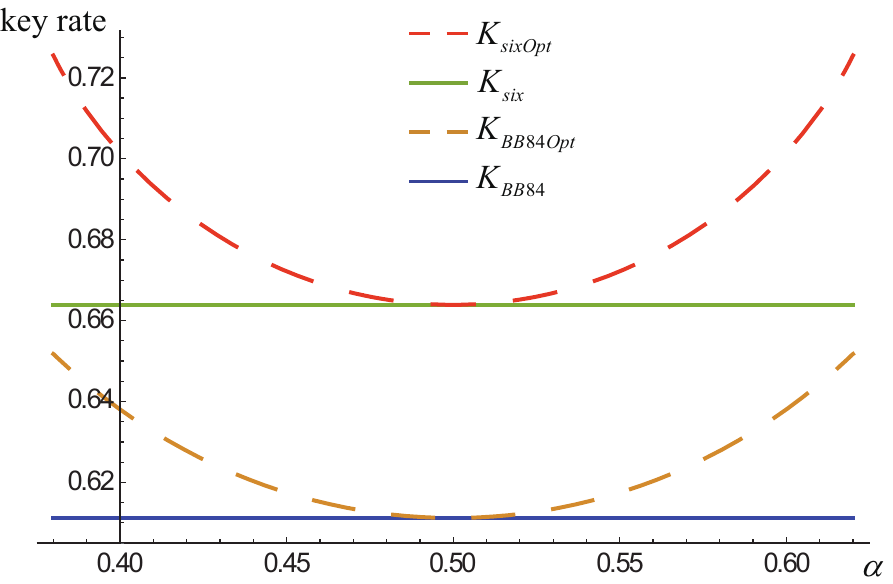}}
\caption{(Color online) Key rate comparison of different schemes. We set $e_x(e_p)=e_y=e_z(e_b)=3\%$. The parameter $\alpha$ describes the unbalance of the diagonal elements, $m_{00}/m_{33}=m_{22}/m_{11}=\frac{\alpha}{1-\alpha}$, where the region $\alpha\in[0.38,0.62]$ is considered to guarantee the non-negativity of the state $\rho_{AB}$. From top to bottom, the key rate plots are $K_{six}^{opt}$, $K_{six}$, $K_{BB84}^{opt}$, and $K_{BB84}$, respectively.}\label{fig:CKR}
\end{figure}

The numerical result shows that the coherence-based key rate of the six-state protocol enjoys the highest key rate, while the Shor-Preskill key rate of BB84 possesses the lowest key rate. As $\alpha=0.5$, that is, there is no unbalance of diagonal elements, $K_{BB84}^{opt}=K_{BB84}$ and $K_{six}^{opt}=K_{six}$; as $\alpha$ departs from $0.5$, the unbalance becomes significant and one can clearly see the improvements on the key rates.

We remark that the unbalance of the diagonal elements could happen in practical QKD scenarios. In the next section, one can see that the asymmetry of the detectors can lead to this phenomenon [see $\rho_{AB}^Z$ in Eq.~\eqref{Zmatrix:general} for an example].
\section{Practical issue:  detection efficiency mismatch}\label{Sec:mismatch}
In this section, we apply our coherence framework to QKD security analysis when considering a practical issue --- detection efficiency mismatch. Here, we focus on analyzing the BB84 protocol. We show that the key rate derived by our framework is generally higher than in the previous analyses~\cite{fung2009mismatch}.

Ideally, the two detectors which detect $\ket{0}$ and $\ket{1}$ in $Z$ basis ($\ket{+}$ and $\ket{-}$ in $X$ basis) respectively are assumed to be identical. However, in practical scenarios, there are always imperfections in the channels and detectors, which may lead to different efficiencies for $\ket{0}$ and $\ket{1}$ (or $\ket{+}$ and $\ket{-})$~\cite{zhao2008quantum}.

\subsection{Detector model}
In practice, the detection efficiency of a detector is normally related to other degrees of freedom of the input photons, such as time, space, or spectrum~\cite{fung2009mismatch}. For example, Fig.~\ref{fig:Dmis} shows the detection efficiency mismatch that is related to the temporal degree of freedom of the injected photons. Employing the analytical methods in Ref.~\cite{fung2009mismatch}, here we model the measurement by the two detectors on Bob's side by the measurement of
\begin{eqnarray}
&M_0=\eta_0\ket{0}_B\bra{0},\\
&M_1=\eta_1\ket{1}_B\bra{1},
\end{eqnarray}
where $0\leq \eta_0$, $\eta_1 \leq 1$ are the efficiencies of the two detectors. We assume $\eta_0$ and $\eta_1$ can be calibrated thus are known to Alice and Bob.

Here, we decompose $M_0$ and $M_1$ by a filtering operation $F_z$ and an ideal $Z$-basis measurement, where
\begin{equation}\label{Eq:filterZ}
F_z=\sqrt{\eta_0}\ket{0}_B\bra{0}+\sqrt{\eta_1}\ket{1}_B\bra{1}.
\end{equation}
Similarly, the measurement in the $X$ basis with the two nonidentical detectors can be decomposed by a filtering operation,
\begin{equation}
F_x=\sqrt{\eta_0}\ket{+}_B\bra{+}+\sqrt{\eta_1}\ket{-}_B\bra{-}
\end{equation}
followed by an ideal $X$-basis measurement $\{\ket{+}_B\bra{+}, \ket{-}_B\bra{-}\}$.

Under this decomposition, before the ideal $Z$-basis measurement, the state is transformed to
\begin{equation}\label{Eq:Zstate}
\rho_{AB}^Z=\frac{F_z\rho_{AB}F_z}{\mathrm{Tr}(F_z\rho_{AB} F_z)},
\end{equation}
and the obtained bit error rate is represented by
\begin{equation}\label{Eq:eb}
e_b=\mathrm{Tr}(\Pi^-\rho_{AB}^Z).
\end{equation}
Similarly, for the $X$-basis measurement, one has
\begin{equation}\label{Eq:Xstate}
\rho_{AB}^X=\frac{F_x\rho_{AB}F_x}{\mathrm{Tr}(F_x\rho_{AB} F_x)},
\end{equation}
and the obtained phase error rate is
\begin{equation}\label{Xphase:error}
e_p=\mathrm{Tr}(\Pi^-_x\rho_{AB}^X),
\end{equation}
where $\Pi^-_x=\ket{+-}\bra{+-}+\ket{-+}\bra{-+}$. We remark that $e_p$ is \emph{not} the phase error corresponding to the state measured in the $Z$ basis. That error should be
\begin{equation}
e_p'=\mathrm{Tr}(\Pi^-_x\rho_{AB}^Z).
\end{equation}
Note that the discrepancy between $e_p$ and $e_p'$ originates from the detection efficiency mismatch of the two detectors.

\begin{figure}[bht]
\centering
\resizebox{8cm}{!}{\includegraphics[scale=1]{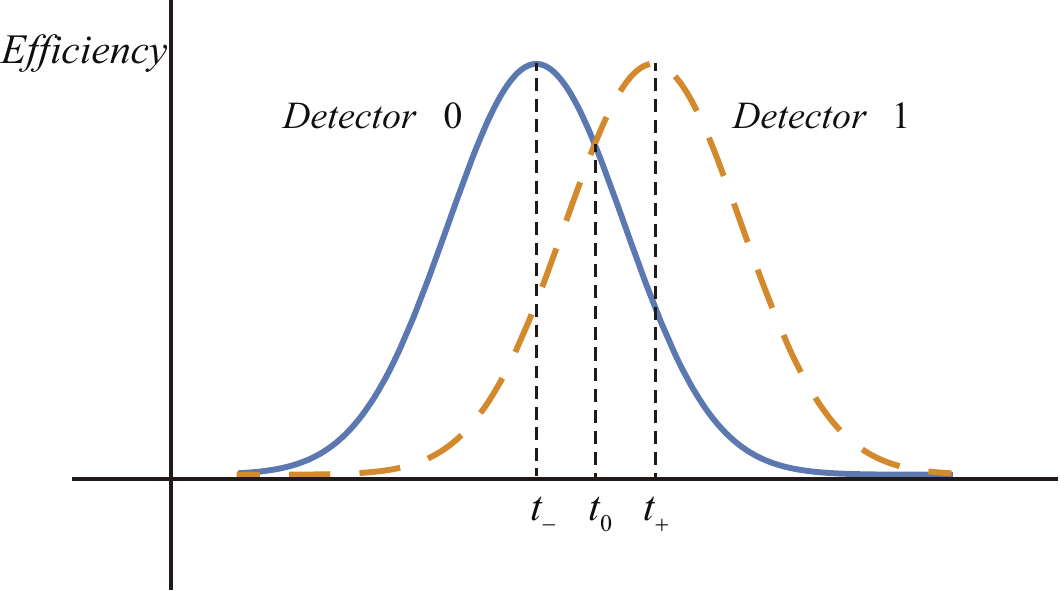}}
\caption{(Color online) Illustration of detector efficiency mismatch in the time domain. Due to the optical path difference between the two detectors, the two detectors have different detection efficiency in the time domain. If the arrival time of the signal is $t_0$, the efficiencies of the two detectors are the same. However, if the arrival time is $t_-$ or $t_+$, the efficiency of Detector 0 (the solid blue line) is higher or lower.}\label{fig:Dmis}
\end{figure}

\subsection{Derivation of the key rate}
Essentially, the task of deriving the final key rate is to estimate $e_p'$ with the knowledge of the measurement results in the $Z$ and $X$ bases.

Let us explicitly write down $\rho_{AB}^Z$ in Eq.~\eqref{Eq:Zstate},
\begin{equation}\label{Zmatrix:general}
\rho_{AB}^Z=
\frac{1}{\Gamma}\left(
  \begin{array}{cccc}
    \eta_0 m_{00} & \cdot & \cdot &\sqrt{\eta_0\eta_1}m_{03}  \\
    \cdot & \eta_1 m_{11} & \sqrt{\eta_0\eta_1}m_{12} & \cdot \\
    \cdot & \sqrt{\eta_0\eta_1}m_{21} & \eta_0 m_{22} & \cdot \\
    \sqrt{\eta_0\eta_1}m_{30} & \cdot & \cdot & \eta_1 m_{33} \\
  \end{array}
\right),
\end{equation}
where the matrix elements that are not related to the parameter estimation are represented by ``$\cdot$", and the normalization factor is
\begin{equation}\label{Z:diagnormal}
\Gamma=\eta_0 m_{00}+\eta_1 m_{11}+\eta_0 m_{22}+\eta_1 m_{33}.
\end{equation}
Employing Eq.~\eqref{perror:BB84}, one has
\begin{eqnarray}\label{epZ:epAB}
e_p'&=&1/2-\frac{\sqrt{\eta_0\eta_1}}{\Gamma}\{\re[m_{03}]+\re[m_{12}]\},\\\nonumber
&=&1/2-\frac{\sqrt{\eta_0\eta_1}}{\Gamma}(1/2-e_p''),
\end{eqnarray}
where the second line applies Eq.~\eqref{perror:BB84} for $e_p''$, and $e_p''$ is defined as
\begin{equation}\label{ep:AB}
e_p''=\mathrm{Tr}(\Pi^-_x\rho_{AB}).
\end{equation}

Actually, $\Gamma$ and $e_p''$ in Eq.~\eqref{epZ:epAB} both can be obtained from measurement results. Denote the probabilities of obtaining $\ket{00}$, $\ket{01}$, $\ket{10}$, and $\ket{11}$ when both sides are measured in the $Z$ basis by $\hat{m}_{00}, \hat{m}_{11}, \hat{m}_{22}$ and $\hat{m}_{33}$, respectively. Then from Eq.~\eqref{Zmatrix:general} one has
\begin{eqnarray}\label{Z:diag:m}
\hat{m}_{00}=\eta_0 m_{00}/\Gamma\ , \ \ \  \hat{m}_{11}=\eta_1 m_{11}/\Gamma\ ,\\\nonumber
\hat{m}_{22}=\eta_0 m_{22}/\Gamma\ , \ \ \  \hat{m}_{33}=\eta_1 m_{33}/\Gamma\ ,
\end{eqnarray}
Since $m_{00}+m_{11}+m_{22}+m_{33}=1$, $\Gamma$ can be represented as
\begin{equation}\label{Z:diagsum}
\Gamma=\frac{1}{\hat{m}_{00}/\eta_0+\hat{m}_{11}/\eta_1+\hat{m}_{22}/\eta_0+\hat{m}_{33}/\eta_1}.
\end{equation}

Similarly, one can explicitly write down $\rho_{AB}^X$ in Eq.~\eqref{Eq:Xstate} in the X basis,
\begin{equation}\label{Xmatrix:general}
\rho_{AB}^X=
\frac{1}{\Gamma'}\left(
  \begin{array}{cccc}
    \eta_0 m'_{00} & \cdot & \cdot &\sqrt{\eta_0\eta_1}m'_{03}  \\
    \cdot & \eta_1 m'_{11} & \sqrt{\eta_0\eta_1}m'_{12} & \cdot \\
    \cdot & \sqrt{\eta_0\eta_1}m'_{21} & \eta_0 m'_{22} & \cdot \\
    \sqrt{\eta_0\eta_1}m'_{30} & \cdot & \cdot & \eta_1 m'_{33} \\
  \end{array}
\right),
\end{equation}
where $m_{i,j}'$ denotes the matrix elements of $\rho_{AB}$ in the $X$ basis and the normalization factor is
\begin{equation}\label{X:diagnormal}
\Gamma'=\eta_0 m'_{00}+\eta_1 m'_{11}+\eta_0 m'_{22}+\eta_0 m'_{22}.
\end{equation}
Denote the probabilities of obtaining $\ket{++}$, $\ket{+-}$, $\ket{-+}$, and $\ket{--}$ when both sides are measured in the $X$ basis by $\hat{m}'_{00}, \hat{m}'_{11}, \hat{m}'_{22}$ and $\hat{m}'_{33}$, respectively. Then one has
\begin{eqnarray}\label{X:diag}
\hat{m}'_{00}=\eta_0 m'_{00}/\Gamma'\ , \ \ \  \hat{m}'_{11}=\eta_1 m'_{11}/\Gamma'\ ,\\\nonumber
\hat{m}'_{22}=\eta_0 m'_{22}/\Gamma'\ , \ \ \  \hat{m}'_{33}=\eta_1 m'_{33}/\Gamma'\ .
\end{eqnarray}
Similarly to $\Gamma$ in Eq.~\eqref{Z:diagsum} for the Z basis, $\Gamma'$ can be represented as
\begin{eqnarray}\label{Eq:Xdiagsum}
\Gamma'=\frac{1}{\hat{m}'_{00}/\eta_0+\hat{m}'_{11}/\eta_1+\hat{m}'_{22}/\eta_0+\hat{m}'_{33}/\eta_1}.
\end{eqnarray}
By the definition in Eq.~\eqref{ep:AB}, we have
\begin{eqnarray}\label{phaseAB:exp}
e_p''&=&m'_{11}+m'_{22}\\\nonumber
&=&\Gamma'(\hat{m}'_{11}/\eta_1+\hat{m}'_{22}/\eta_0).
\end{eqnarray}
%

With Eqs.~\eqref{epZ:epAB}, ~\eqref{Z:diagsum}, ~\eqref{Eq:Xdiagsum} and  ~\eqref{phaseAB:exp}, the phase error $e_p'$ can be precisely estimated from the measurement results in $Z$ and $X$ bases. By contrast, $e_p'$ is roughly upper bounded in previous results ~\cite{fung2009mismatch}. This precise estimation of $e_p'$ allows Alice and Bob to obtain a higher key rate than in the previous analysis. Also, the key rate can be further improved by applying Theorem \ref{Th:KR:BB84Opt} to $\rho_{AB}^Z$.

Therefore, with fine-grained parameters, one can expect a higher key rate than the previous ones. This is to be illustrated in the following subsection.

\subsection{Analytical key rate formula under symmetric attack}
To simplify the analysis, we assume Eve's attack to be symmetric between bits $0$ and $1$ in the $Z/X$-basis, i.e., the diagonal elements of $\rho_{AB}$ in both the $Z$ basis and the $X$ basis are balanced. That is
\begin{eqnarray}\label{Z:sym}
&m_{00}=m_{33}=c,\\\nonumber
&m_{11}=m_{22}=d,
\end{eqnarray}
with the normalization condition $2(c+d)=1$.
Meanwhile, for the $X$ basis, one has
\begin{eqnarray}\label{X:sym}
&m'_{00}=m'_{33}=c',\\\nonumber
&m'_{11}=m'_{22}=d',
\end{eqnarray}
with $2(c'+d')=1$. Then via Eq.~\eqref{Z:diagnormal}, one has
\begin{eqnarray}\label{Gam:nomea}
\Gamma&=&\eta_0 c+\eta_1 d+\eta_0 d+\eta_1 c\\\nonumber
&=&(\eta_0+\eta_1)(c+d)\\\nonumber
&=&(\eta_0+\eta_1)/2,
\end{eqnarray}
where $\Gamma$ is a constant related to the detection efficiency. 

With Eqs.~\eqref{Xphase:error}, \eqref{Xmatrix:general} and \eqref{X:diagnormal}, one has
\begin{eqnarray}\label{Phase:XeqIni}
e_p&=&(\eta_1 m'_{11}+\eta_0 m'_{22})/\Gamma'\\\nonumber
&=&\frac{\eta_1 m'_{11} + \eta_0 m'_{22}}
{\eta_0 m'_{00} + \eta_1 m'_{11} +\eta_0 m'_{10} + \eta_1 m'_{22}}\\\nonumber
&=&\frac{(\eta_0+\eta_1)d'}{(\eta_0+\eta_1)(c'+d')}\\\nonumber
&=&e_p'',
\end{eqnarray}
where the last line is on account of the definition of $e_p''$.
Inserting Eqs.~\eqref{Gam:nomea} and \eqref{Phase:XeqIni} into Eq.~\eqref{epZ:epAB}, we have
\begin{equation}\label{sym:epZ:epAB}
e_p'=1/2-\frac{2\sqrt{\eta_0\eta_1}}{\eta_0+\eta_1}(1/2-e_p),
\end{equation}
which means $e_p'$ can be precisely estimated with $e_p$.

With Eq.~\eqref{sym:epZ:epAB}, one can estimate the key rate by applying Theorem \ref{Th:KR:BB84Opt}. For the current scenario that is restricted to the symmetric attack, the optimization Problem \ref{Pro:1} can be solved analytically. See Appendix \ref{Der:kmismath} for a detailed derivation. The key rate is given by
\begin{equation}\label{Krate_mismatch}
K=H(x)-H(f(x,e_p))-H(e_b),
\end{equation}
where
\begin{eqnarray*}
&x=\frac{\eta_0}{\eta_0+\eta_1},\\\nonumber
&f(x,e_p)=1/2+\sqrt{(1/2-x)^2+x(1-x)(1-2e_p)^2}.
\end{eqnarray*}

Comparatively, Ref.~\cite{fung2009mismatch} proposes two methods of analyzing the key rate with the detection efficiency mismatch issue. There, one key rate formula is obtained via the data discarding process,
\begin{equation}
K_1=\frac{2min\{\eta_0, \eta_1\}}{\eta_0+\eta_1}\Big(1-H(e_p)-H(e_b)\Big).
\end{equation}
The other key rate is obtained via a virtual protocol based on Koashi's complimentary approach~\cite{Koashi2009simple},
\begin{equation}
K_2=\frac{2min\{\eta_0, \eta_1\}}{\eta_0+\eta_1}\Big(1-H(e_p)\Big)-H(e_b).
\end{equation}

To compare above key rates, $K$, $K_1$ and $K_2$, we first consider the \textit{noiseless} case, where $e_p=e_b=0$. Then, one has $K=H(\frac{\eta_0}{\eta_0+\eta_1})$ and $K_1=K_2=\frac{2min\{\eta_0, \eta_1\}}{\eta_0+\eta_1}$. It is clear that $K\geq K_1 = K_2$. 
In the \textit{noisy} case, the key rates obtained from the three analysis are plotted in Fig.~\ref{Fig:mismatch}. It shows that $K$ is larger than $K_2$ for any efficiency mismatch extent; 
while $K$ is larger than $K_1$ if the mismatch is not too serious. This manifests the advantage of coherence framework for analyzing QKD security.

When the efficiency mismatch becomes large ($x$ approaches 0 in Fig.~\ref{Fig:mismatch}), $K$ becomes negative; but $K_1$ keeps positive and thus larger than $K$. This fact can be understood as follows. Suppose the initial state before measurement $\rho_{AB}$ possesses positive key rate (bit and phase error are both small). The data discarding approach effectively transforms the state $\rho_{AB}^Z$ to $\rho_{AB}$ with probability $\frac{2min\{\eta_0, \eta_1\}}{\eta_0+\eta_1}$, thus the key rate $K_1$ is always positive. As $x\rightarrow 0$ $(\eta_0\rightarrow 0$), the probability of successful transform approaches $0$, thus $K_1 \rightarrow 0$. On the other hand, as $x\rightarrow 0$, the phase error of the state $\rho_{AB}^Z$  in Eq.~\eqref{sym:epZ:epAB} approaches $1/2$. 
Consequently, the first two terms in Eq.\eqref{Krate_mismatch} approaches zero, and $K \rightarrow-H(e_b)\leq 0$ .

\begin{figure}[bht]
\centering
\resizebox{8cm}{!}{\includegraphics[scale=1]{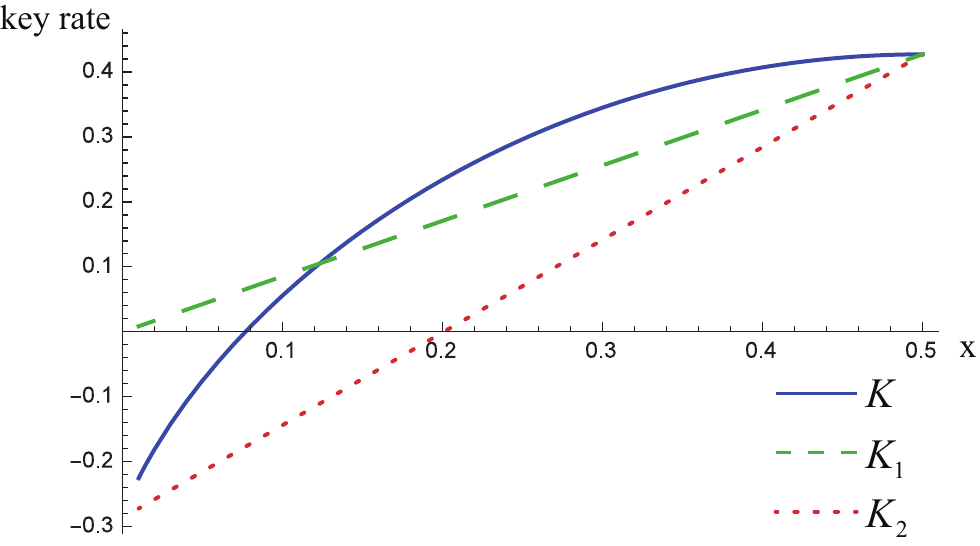}}
\caption{(Color online) Comparison between the key rates obtained from three analyses targeting the detection efficiency mismatch issue. The solid blue, dashed green and dotted red curves represent $K$, $K_1$ and $K_2$, respectively. We set $e_p=e_b=5\%$, and plot the key rates versus  $x=\frac{\eta_0}{\eta_0+\eta_1}$, which describes the efficiency mismatch.  Without loss of generality, we assume $\eta_0\leq\eta_1$ and hence $0\leq x \leq0.5$.
}\label{Fig:mismatch}
\end{figure}


\section{Relation with entanglement} \label{Sec:Entanglement}
The existing security analyses~\cite{Lo1999unconditional,shor2000simple} usually starts from entanglement distillation protocols~\cite{Bennett1996Purification,Horodecki09Quantum}, where entanglement is taken as an essential resource to deliver the security of the final key. In contrast, our work is based on the relation between quantum coherence and intrinsic randomness, which is related to the security in QRNG and QKD. Specifically, after correcting the bit errors, we take Alice and Bob as a whole and analyze the intrinsic randomness out of reach of Eve. From this point of view, our approach shares similarity with Koashi's, which is based on the complementary arguments for the joint system~\cite{Koashi2009simple}. In addition, recently there have been several works considering the interplay between coherence and entanglement~\cite{Ma2016converting,Streltsov2015measuring}. However, we remark that in these works the authors normally investigate converting subsystem coherence (not the coherence in the bipartite system) into global correlations and the incoherent operations used there cannot be directly applied into the QKD security analysis.

Here, we show some relations between our key rate and the entanglement property of the input state $\rho_{AB}$. As shown in Eq.~\eqref{Eq:KReb}, the key rate can be enhanced if Alice and Bob acquire more information of the shared state $\rho_{AB}$ and estimate the coherence of $\Phi(\rho_{AB})$ more accurately. Suppose Alice and Bob perform a full tomography of $\rho_{AB}$ in the parameter estimation step, then the state set $S$ only contains one state $\rho_{AB}$. 
An upper bound for the key rate is shown in the following proposition.
\begin{proposition}\label{Ob:hash}
Consider a protocol in which a full tomography on $\rho_{AB}$ is performed in the parameter estimation, then the key rate is upper bounded by
\begin{equation}\label{Eq:hash}
K(\rho_{AB})=C(\Phi(\rho_{AB}))-H(e_b)\leq S(\mathrm{Tr}_A(\Phi(\rho_{AB})))-S(\Phi(\rho_{AB})).
\end{equation}
\end{proposition}
Note that the right side of Eq.~\eqref{Eq:hash} is the hashing inequality for the state $\Phi(\rho_{AB})$, which is a lower bound for one-way LOCC (local operations and classical communication) entanglement distillation protocol~\cite{Devetak2005distillation,Horodecki09Quantum}. We remark that the projection operation $\Phi$ on state $\rho_{AB}$ is a nonlocal operation, hence our analysis framework based on coherence could potentially yield higher key rate than the usual entanglement distillation analysis. For conciseness, we leave the proof of Proposition \ref{Ob:hash} in Appendix \ref{Ap:profhash}. We also compare our key rate with the Devetak-Winter formula~\cite{Devetak2005distillation} in Appendix \ref{Ap:Winter}.

Now we define a key rate that is independent of measurement basis by maximizing the key rate generated by state $\rho_{AB}$ over all local bases, i.e.,
\begin{equation}\label{}
K ^{m}(\rho_{AB})= \max_{Z_A\otimes Z_B}\Big\{ C(\Phi(\rho_{AB}))-H(e_b) \Big\},
\end{equation}
where $Z_A\otimes Z_B$ labels all the local bases of Alice and Bob.

One can see that given a pure state $\ket{\Psi_{AB}}$, the maximal key rate is its entanglement entropy, i.e., $K ^{m}(\Psi_{AB})=S(\rho_A)$, with $\Psi_{AB}=\ket{\Psi_{AB}}\bra{\Psi_{AB}}$ and $\rho_A=\mathrm{Tr}_B(\Psi_{AB})$. To be specific, suppose $\ket{\Psi}_{AB}=a_0\ket{\psi_0}_A\ket{\psi'_0}_B+a_1\ket{\psi_1}_A\ket{\psi'_1}_B$ is the the Schmidt decomposition of $\ket{\Psi}_{AB}$, one can choose $\{\ket{\psi_0}_A,\ket{\psi_1}_A\}$ and $\{\ket{\psi'_0}_B,\ket{\psi'_1}_B\}$ as the optimal local basis that maximizes the key rate. In addition, it is clear for a product state $\ket{\Psi_{AB}}=\ket{\psi_0}_A\ket{\psi'_0}_B$ the maximal key rate is zero. And the following proposition gives an upper bound for the maximal key rate for general state.

\begin{proposition}
The maximal key rate of the state $\rho_{AB}$ optimized over all local bases is upper bounded by the entanglement of formation,
\begin{equation}\label{Eformation}
K ^{m}(\rho_{AB})\leq \min_{\{p_i, \Psi_i\}}\sum_i p_i K^m(\Psi_i)=E^{form}(\rho_{AB}),
\end{equation}
where the minimization is over all the convex decompositions of $\rho_{AB}=\sum p_i \Psi_i$.
\end{proposition}
\begin{proof}
Note that the key rate $K$ is convex due to the convexity of the relative entropy of coherence and the concavity of the von Neumann entropy. Hence, for any decomposition of $\rho_{AB}=\sum p_i \Psi_i$,
\begin{equation}\label{}
K ^{m}(\rho_{AB})= K(\rho_{AB})_{Z_A^o\otimes Z_B^o}\leq  \sum_i p_i K(\Psi_i)_{Z_A^o\otimes Z_B^o} \leq \sum_i p_i K^m(\Psi_i),
\end{equation}
where $Z_A^o\otimes Z_B^o$ represents the optimal local basis for $\rho_{AB}$ and the last inequality holds since one can improve the key rate of $\Psi_i$ further by choosing specific optimal basis for each of them. Consequently, the maximal key rate is upper bounded by the entanglement of formation as shown in Eq.~\eqref{Eformation}.
\end{proof}
It is also clear to see that the key rate for any separable state $K ^{m}(\rho_{AB}^{sep})\leq 0$, since it can be written as the combination of product states~\cite{Curty2004Precondition}.

\section{Discussion and conclusion}\label{Sec:conclude}
We have proposed a framework that captures the close relation between coherence and QKD. By considering a generic entanglement-based QKD protocol, the framework shows that the secure key rate can be quantified via the coherence of the shared bipartite quantum states. By applying the proposed framework to the BB84 and six-state protocols, we reproduce the key rates. Furthermore, the framework can even allow us to improve the key rates by modifying the postprocessing. And it is also shown to be advantageous to analyze the practical issue of detector efficiency mismatch in QKD. More generally, the coherence-based framework also provide us convenience to analyze the key rate by using tools from coherence theory. 

Along this direction, a number of problems can be explored in the future. Note that our current security analysis is performed under the collective attack scenario. An important future study is to extend the results to the scenario of coherent attacks and take into account finite-key effects~\cite{Renner2005Security,Ma2011Finite}. To solve this problem, one may employ results from recent studies on one-shot coherence theories~\cite{Zhao2018Dilution,Zhao2018On,Bartosz2018One}. There, coherence can be quantified in a non-asymptotic setting where finite data-size effects appear. Also, apart from the currently studied cases, the framework has potential to be applied to many other QKD protocols, such as the three-state protocol~\cite{Boileau05Three,Fung2006three} and B92 protocol~\cite{Bennett92Quantum}, where similar derivation and improvement of the key rate are expected. In particular, the framework can be naturally extended to measurement-device-independent QKD~\cite{Lo12MDI}, since bipartite quantum states are directly distributed and measured in this kind of protocol. In addition, generalization to high dimensional QKD and continuous-variable QKD~\cite{scarani2009security} is also interesting. Also, we expect our framework to be useful in addressing more practical issues in QKD, the solutions to many of which are missing or very complicated at the moment.

Moreover, it is intriguing to reexamine the previous QKD security analyses from the coherence theory point of view. To be specific, a common technique of security analysis is to transform the original protocol to the equivalent virtual protocol. The latter is easier to analyze but share the same amount of secure keys. In the virtual protocol, the operations conducted there are incoherent operations~\cite{streltsov2017coherence} (more specifically, dephasing-covariant incoherent operation~\cite{Eric2016Critical,Marvian2016speakable}), which commute with the final $Z$-basis measurement to generate keys. It is also interesting to investigate the connection between coherence and entanglement~\cite{Ma2016converting,Streltsov2015measuring,Patrick2012Unification,Yao15Quantum,Alexander17Towards} under the scenario of security analysis. This topic might not only deepen our understanding of the basic quantum resources, but also inspire useful applications in quantum information processing.

\acknowledgments
J.~Ma and Y.~Zhou contributed equally to this work. We acknowledge H.~Zhou, P.~Zeng, and A.~S.~Trushechkin for the insightful discussions and useful comments. This work was supported by the National Natural Science Foundation of China Grants No.~11674193 and No.~11875173, and the National Key R\&D Program of China Grants No.~2017YFA0303900 and No.~2017YFA0304004.

\appendix
\section{Quantum bit error correction} \label{App:QEC}
We first clarify the information reconciliation of the original protocol in (i)-(v), and then convert it to a quantum version, the quantum bit error correction of the virtual protocol.

In the original protocol, after the $Z$-basis measurement on $\rho_{AB}^{\otimes n}$, Alice and Bob get $n$-bit strings $\mathcal{Z}_A^n$ and $\mathcal{Z}_B^n$, respectively. Due to errors, the random variables $\mathcal{Z}_A^n$ and $\mathcal{Z}_B^n$ are not identical in general. Then, in (linear) error correction, Alice generate an error syndrome by hashing her bit string with an $nH(\mathcal{Z}_A|\mathcal{Z}_B)\times n$ random binary matrix. By consuming $nH(\mathcal{Z}_A|\mathcal{Z}_B)$ preshared secret bits, Alice sends the syndrome to Bob safely with the one-time-pad encryption. After obtaining the syndromes, Bob can correct the corresponding error bits.

In the virtual protocol, the quantum bit error correction is executed before the $Z$-basis measurement. Specifically, Alice and Bob now share $nH(\mathcal{Z}_A|\mathcal{Z}_B)$ EPR pairs. First, they use their state $\rho_{AB}^{\otimes n}$ to control the EPR pairs according to the hashing matrix separately, where the ancillary EPR pairs act as the target of the CNOT gate. Second, they measure the EPR pairs in the $Z$ basis separately and get the measurement results $\mathcal{Z}_A^a$ and $\mathcal{Z}_B^a$, where $a$ labels the ancilla. Then Alice sends $\mathcal{Z}_A^a$ to Bob and Bob obtains the error syndrome via bitwise binary addition $\mathcal{Z}_A^a\oplus \mathcal{Z}_B^a$. Finally, Bob locates the bit errors and applies the $\sigma_x$ operation to correct them. Here, it is clear that the quantum bit error correction commutes with the $Z$-basis measurement. Take a simple example, where
\begin{equation}\label{Eq:HashingM}
H_{2\times 3}=\left(
  \begin{array}{ccc}
    1 & 1 & 0\\
    0 & 1 & 1\\
  \end{array}
\right);
\end{equation}
the circuit is illustrated in Fig.~\ref{Fig:Hashing}.

\begin{figure}[bht]
\centering
\resizebox{8cm}{!}{\includegraphics[scale=1]{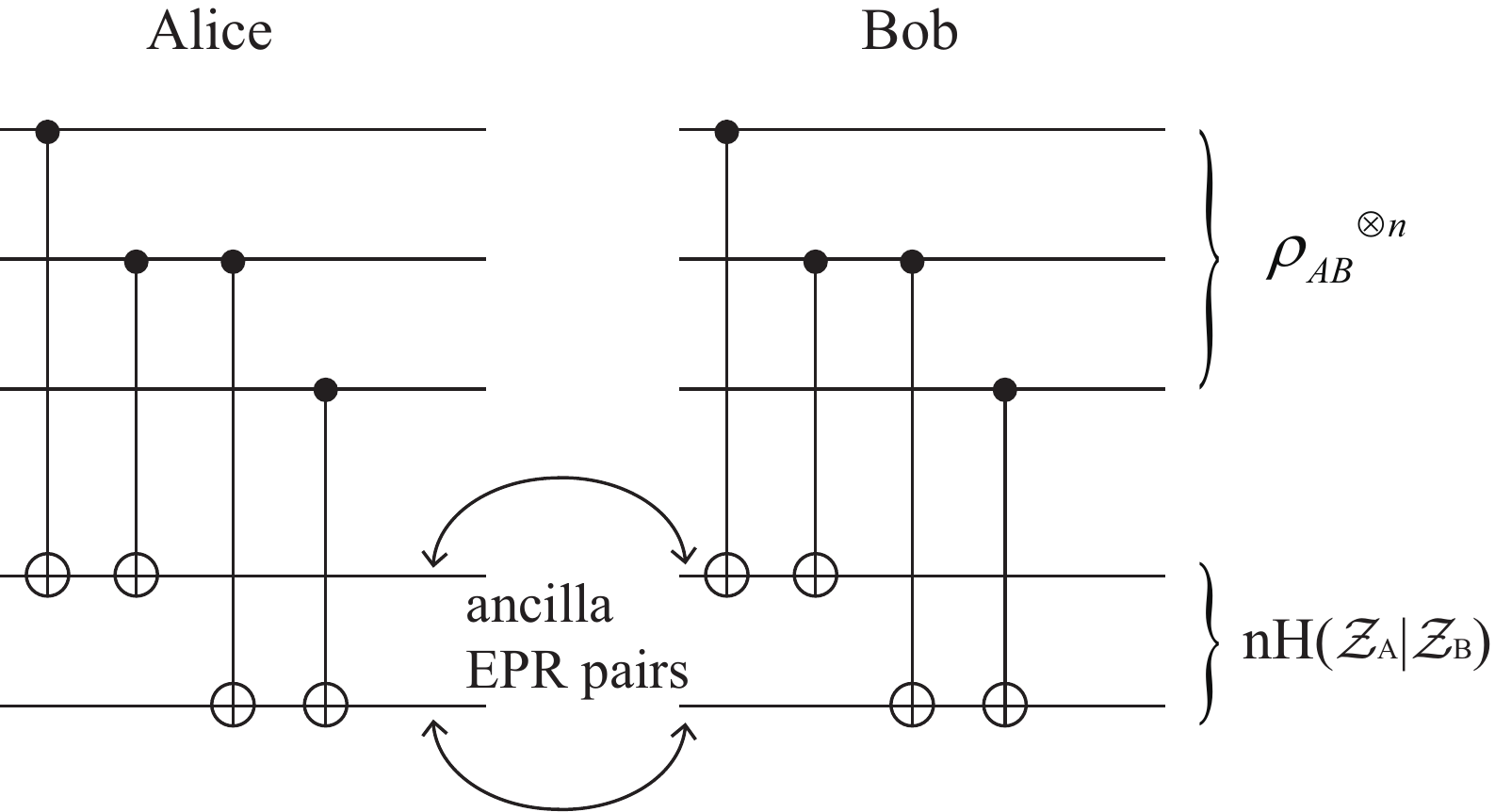}}
\caption{The circuit for quantum bit error correction. Alice and Bob separately use their state $\rho_{AB}^{\otimes n}$ to control the EPR pairs according to a $nH(\mathcal{Z}_A|\mathcal{Z}_B)\times n$ hashing matrix with $n\rightarrow\infty$. Here for clearness we show the schematic with a $2\times 3$ hashing matrix $H_{2\times 3}$ given in Eq.~\eqref{Eq:HashingM}.}\label{Fig:Hashing}
\end{figure}

\section{Analytical solution to Problem \ref{Pro:1}}\label{sym:pro1}
\begin{lemma}\label{Lem:sym:pro1}
If the four diagonal elements of the density matrix in Eq.~\eqref{Eq:Pro:1} satisfy  $m_{00}/m_{33}=m_{11}/m_{22}=\frac{\alpha}{1-\alpha}$ (or $m_{00}/m_{33}=m_{22}/m_{11}$), the minimal coherence obtained from the optimization in Problem \ref{Pro:1} is $H(\alpha)-H\left(\frac1{2}+\sqrt{(\alpha-\frac1{2})^2+(\frac1{2}-e_p)^2}\right)$, with the solution $\bar{a}=(1-e_b)(1/2-e_p)$ and $\bar{b}=e_b(1/2-e_p)$.
\end{lemma}
\begin{proof}
Here, we only consider the case $m_{00}/m_{33}=m_{11}/m_{22}$, and the proof for the other case $\bar{m}_{00}/\bar{m}_{33}=\bar{m}_{22}/\bar{m}_{11}$ can proceed in a similar way.

Due to the additivity property of coherence, we can express the coherence of $\rho(a,b)$, like in Eq.\eqref{Eq:3K12}, as
\begin{equation}\label{AEq:sym:p1}
C(\rho(a,b))=(1-e_b)C(\rho^+)+e_bC(\rho^-)
\end{equation}
where $\rho^{+(-)}=\Pi^{+(-)}\rho(a,b)\Pi^{+(-)}/\mathrm{Tr}(\Pi^{+(-)}\rho(a,b))$. With $m_{00}/m_{33}=m_{11}/m_{22}=\frac{\alpha}{1-\alpha}$, we can explicitly write down the matrix form of $\rho^{+(-)}$ as,

\begin{equation}
\rho^+=
\left(
  \begin{array}{cc}
    \alpha & a'\\
     a' & 1-\alpha
  \end{array}
\right),\ \
\rho^-=
\left(
  \begin{array}{cc}
    \alpha & b'\\
     b' & 1-\alpha
  \end{array}
\right),
\end{equation}
where $(1-e_b)a'=a$ and $e_bb'=b$.

The relative entropy of coherence of the state $\rho^+$ is,
\begin{eqnarray}\label{}
C(\rho^+)&=&S(\Delta(\rho^+))-S(\rho^+),\\\nonumber
&=&H(\alpha)-H\left(\frac1{2}+\sqrt{(\alpha-\frac1{2})^2+{|a'|}^2}\right),
\end{eqnarray}
where $\Delta(\cdot)$ is the dephasing operation of the $\{\ket{00},\ket{11}\}$ basis.
For simplicity, we denote $g(x)=\frac1{2}+\sqrt{\gamma+x^2}$, where $\gamma=(\alpha-\frac1{2})^2$, hence
\begin{equation}\label{Aeq:sym:p}
C(\rho^+)=H(\alpha)-H(g(|a'|)).
\end{equation}
And similarly for $\rho^-$, we have
\begin{equation}\label{Aeq:sym:m}
C(\rho^-)=H(\alpha)-H(g(|b'|)).
\end{equation}

In fact,
$g(x)$ is a monotonically increasing convex function, on account of
\begin{eqnarray}
g'(x)&=\frac{x}{\sqrt{\gamma+x^2}}\geq0,\\\nonumber
g''(x)&=\frac{\gamma}{{(\gamma+x^2)}^{\frac{3}{2}}}>0.
\end{eqnarray}
Moreover we can show that $H\big{(}g(x)\big{)}$ is a concave function. Specifically, for two variables $x_1$, $x_2$ with probability $p_1$, $p_2$,
\begin{equation}\label{}
\sum_ip_iH\big{(}g(x_i)\big{)}\leq H\Big{(}\sum_i p_i g(x_i)\Big{)}\leq H\Big{(}g\Big{(}\sum_i p_ix_i\Big{)}\Big{)},
\end{equation}
where the summation $i=1,2$. The first inequality is due to the concavity of the entropy function $H(x)$. The second inequality holds because of three facts: $g(x)$ is a convex function, i.e., $\sum_i p_i g(x_i)\geq g\Big{(}\sum_i p_ix_i\Big{)}$; $g(x)$ is larger than $\frac{1}{2}$; $H(x)$ monotonically decreases as $x\geq\frac1{2}$.

Then inserting Eq.~\eqref{Aeq:sym:p} and \eqref{Aeq:sym:m} into Eq.~\eqref{AEq:sym:p1}, and utilizing the concavity of $H(g(x))$, we have
\begin{eqnarray}\label{Aeq:eq:coh}
C(\rho(a,b))&=&H(\alpha)-(1-e_b)H(g(|a'|))-e_bH(g(|b'|))\\\nonumber
&\geq& H(\alpha)-H(g([1-e_b]|a'|+e_b|b'|)),
\end{eqnarray}
where the equality holds when $|a'|= |b'|$.

Remembering that the coherence $C(\rho(a,b))$ should be minimized under the constraint $a+b=(1-e_b)a'+e_bb'=\frac1{2}-e_p$,  we have
\begin{equation}\label{}
(1-e_b)|a'|+e_b|b'|\geq|(1-e_b)a'+e_bb'|=|\frac1{2}-e_p|,
\end{equation}
where the equality is saturated when $a'$ and $b'$ share the same sign. Consequently, following Eq.~\eqref{Aeq:eq:coh}, we have
\begin{eqnarray}\label{}
C(\rho(a,b))&\geq& H(\alpha)-H(g([1-e_b]|a'|+e_b|b'|))\\\nonumber
&\geq& H(\alpha)-H(g(|\frac1{2}-e_p|)),
\end{eqnarray}
where the second inequality holds since $H(g(x))$ is a monotonically decreasing function. And the inequality is saturated when $\bar{a}=(1-e_b)(1/2-e_p)$ and $\bar{b}=e_b(1/2-e_p)$.
\end{proof}
\section{Derivation of the key rate $K$ in Eq.~\eqref{Krate_mismatch}}\label{Der:kmismath}
Here we derive the key rate for the symmetric attack scenario, where the key is generated from the $Z$-basis bit of $\rho_{AB}^Z$.  Under the symmetric assumption in Eq.~\eqref{Z:sym},
the four diagonal elements in the $Z$-basis of $\rho_{AB}^Z$ are, $2c\frac{\eta_0}{\eta_0+\eta_1}$, $2d\frac{\eta_1}{\eta_0+\eta_1}$, $2d\frac{\eta_0}{\eta_0+\eta_1}$, and $2c\frac{\eta_1}{\eta_0+\eta_1}$ respectively. And the phase error $e_p'$ is given by Eq.~\eqref{sym:epZ:epAB}. Consequently, the coherence minimization of Problem \ref{Pro:1} becomes
\begin{equation}
\rho(a,b)=
\left(
  \begin{array}{cccc}
    2c\frac{\eta_0}{\eta_0+\eta_1} & 0 & 0 & a\\
    0 & 2d\frac{\eta_1}{\eta_0+\eta_1} & b & 0 \\
    0 & b & 2d\frac{\eta_0}{\eta_0+\eta_1} & 0 \\
    a & 0 & 0 & 2c\frac{\eta_1}{\eta_0+\eta_1} \\
  \end{array}
\right),
\end{equation}
where $a+b=\frac1{2}-e_p'$. It is clear to find that this minimization satisfies the condition in Lemma \ref{Lem:sym:pro1} with $\alpha=\frac{\eta_0}{\eta_0+\eta_1}$. Hence, according to Theorem \ref{Th:KR:BB84Opt}, the key rate reads

\begin{eqnarray}
K
&=&H(\alpha)-H\left(1/2+\sqrt{(\alpha-\frac1{2})^2+(\frac1{2}-e_p')^2}\right)-H(e_b)\\\nonumber
&=&H(\alpha)-H\left(1/2+\sqrt{(\alpha-\frac1{2})^2+\alpha(1-\alpha)(1-2e_p)^2}\right)-H(e_b),
\end{eqnarray}
where Eq.~\eqref{sym:epZ:epAB} is applied in the third line to express $e_p'$ with $e_p$. If we substitute $x$ for $\alpha$ in the above equation, we get the key rate in Eq.~\eqref{Krate_mismatch} in the main part.

%
\section{Proof of Proposition \ref{Ob:hash}}\label{Ap:profhash}
From Eq.~\eqref{Eq:3K12}, one has $C(\Phi(\rho_{AB}))=(1-e_b)C(\rho_{AB}^+)+e_b C(\rho_{AB}^-)$. By definition~\cite{Plenio2014coherence},
\begin{eqnarray}\label{}
C(\rho_{AB}^+)&=&S(\rho_{AB}^{+\textrm{diag}})-S(\rho_{AB}^+)\\\nonumber
&=&LS(\rho_{B}^+)-S(\rho_{AB}^+),
\end{eqnarray}
where $\rho_{B}^+=\mathrm{Tr}_B(\rho_{AB}^+)$. Here in the second line we utilize the fact that $S(\rho_{B}^+)=S(\rho_{AB}^{+\textrm{diag}})$, since $\rho_{AB}^+$ is in the $\Pi^+$ subspace. Similarly, one has
\begin{eqnarray}\label{}
C(\rho_{AB}^-)&=S(\rho_{B}^-)-S(\rho_{AB}^-).
\end{eqnarray}
As a result,
\begin{eqnarray}\label{}
K(\rho_{AB})&=&C(\Phi(\rho_{AB}))-H(e_b)\\\nonumber
&=&(1-e_b)\Big(S(\rho_{B}^+)-S(\rho_{AB}^+)\Big)+e_b\Big(S(\rho_{B}^-)-S(\rho_{AB}^-)\Big)-H(e_b)\\\nonumber
&=&(1-e_b)S(\rho_{B}^+)+e_bS(\rho_{B}^-)-\Big((1-e_b)S(\rho_{AB}^+)+e_b S(\rho_{AB}^-)+  H(e_b)\Big) ,\\\nonumber
&\leq& S((1-e_b)\rho_{B}^++e_b\rho_{B}^-)-S(\Phi(\rho_{AB})),\\\nonumber
&=& S\Big((1-e_b)\mathrm{Tr}_A(\rho_{AB}^+)+e_b\mathrm{Tr}_A(\rho_{AB}^-)\Big)-S(\Phi(\rho_{AB})),\\\nonumber
&=&S(\mathrm{Tr}_A(\Phi(\rho_{AB})))-S(\Phi(\rho_{AB})),
\end{eqnarray}
where the inequality in the fourth line is due to the concavity of entropy.

\section{Comparison with the Devetak-Winter formula}\label{Ap:Winter}
The Devetak-Winter formula shows that the key rate of state $\rho_{AB}$ in the i.i.d.~scenario is $K_{D-W}=S(\mathcal{Z}_A|E)-H(\mathcal{Z}_A|\mathcal{Z}_B)$.
This formula considers the one-way information reconciliation protocol. And in this case, the information reconciliation term $H(\mathcal{Z}_A|\mathcal{Z}_B)$ is the same as in our formula. Note that our formula Eq.~\eqref{Th:keyrate} can be applied to more general information reconciliation protocols, whereas the Devetak-Winter one is originally designed for one-way postprocessing. Thus, we focus on the first term $S(\mathcal{Z}_A|E)$ which is used to estimate the privacy of the sifted key on Alice's side. In fact, it can be written in the relative entropy form~\cite{Patrick2012Unification,Patrick2016description} as
\begin{equation}\label{Eq:winter}
S(\mathcal{Z}_A|E)=D(\rho_{AB}\|\Delta_{Z_A}(\rho_{AB})),
\end{equation}
where $\Delta_{Z_A}$ is the \emph{partial} dephasing operation on system $A$, i.e., $\Delta_{Z_A}(\rho_{AB})=\sum_{i=0,1}\ket{i}_A\bra{i}\rho_{AB}\ket{i}_A\bra{i}$. Here $S(\mathcal{Z}_A|E)$ equals to the amount of basis-dependent \emph{discord} of $\rho_{AB}$~\cite{Modi2012Discord}.

On the other hand, the term corresponding to privacy, $C(\Phi(\rho_{AB}))$ in our key formula in Eq.~\eqref{Eq:KR}, can also be written in the relative entropy form. By definition~\cite{Plenio2014coherence}, we have
\begin{equation}\label{}
C(\Phi(\rho_{AB}))=D(\Phi(\rho_{AB})\|\Delta_{Z_{AB}}(\Phi(\rho_{AB}))).
\end{equation}
Compared with Eq.~\eqref{Eq:winter} of Devetak-Winter formula, $C(\Phi(\rho_{AB}))$ quantifies the global coherence of $\Phi(\rho_{AB})$.

It is enlightening to note that using the same fine-grained parameters, one can achieve the same key rate improvement from the Devetak-Winter formula as our coherence framework. Here is the proof. As shown in Eq.~\eqref{Eq:Pro:1}, $\rho_{AB}$ constrained by the fine-grained parameters in the BB84 protocol satisfies $\rho_{AB}=\Phi(\rho_{AB})$. Similarly, Eq.~\eqref{Eq:Tau} shows that $\rho_{AB}=\Phi(\rho_{AB})$ is also satisfied for $\rho_{AB}$ constrained by the fine-grained parameters in the six-state protocol. Therefore, for both protocols, one has
\begin{eqnarray}\label{Eq:DWcoh}
C(\Phi(\rho_{AB}))&=&D(\Phi(\rho_{AB})\|\Delta_{Z_{AB}}(\Phi(\rho_{AB})))\\\nonumber
&=&D(\Phi(\rho_{AB})\|\Delta_{Z_A}(\Phi(\rho_{AB})))\\\nonumber
&=&D(\rho_{AB}\|\Delta_{Z_A}(\rho_{AB}))\\\nonumber
&=&S(\mathcal{Z}_A|E)
\end{eqnarray}
where the second equality employs the fact that $\Delta_{Z_A}(\Phi(\cdot))=\Delta_{Z_{AB}}(\Phi(\cdot))$ and the third equality employs $\rho_{AB}=\Phi(\rho_{AB})$.

Therefore, with the fine-grained parameters, $K_{D-W}$ is equal to the key rate formula \eqref{Th:keyrate}. This implies one can derive the same improved key rate formulas as those in Theorem \ref{Th:KR:BB84Opt} and Theorem \ref{Th:KR:sixOpt}, from the Devetak-Winter formula with fine-grained parameters.

\bibliographystyle{apsrev4-1}

\bibliographystyle{apsrev4-1}

\bibliography{bibCoherenceQKD}

\end{document}